\documentclass[10pt]{article}
\usepackage{graphicx} 
\usepackage{amsfonts}
\usepackage{amsmath}
\usepackage{amsthm}
\usepackage[margin=1in]{geometry}
\usepackage{bbm}
\usepackage{xcolor}
\usepackage{comment}
\usepackage{hyperref}
\usepackage{natbib}
\usepackage[noend, ruled]{algorithm2e}

\usepackage[utf8]{inputenc} 
\usepackage[T1]{fontenc}    
\usepackage{hyperref}       
\usepackage{url}            
\usepackage{booktabs}       
\usepackage{amsfonts}       
\usepackage{nicefrac}       
\usepackage{microtype}      
\usepackage{xcolor}         

\usepackage{amsmath}
\usepackage{amsthm}
\usepackage{bbm}
\usepackage{enumitem}

\newcommand{\N}{\mathbb{N}}

\newcommand{\E}{\mathbb{E}}
\newcommand{\R}{\mathbb{R}}
\newcommand{\ALG}{\textsf{ALG}}
\newcommand{\CR}{\textsf{CR}}

\newcommand{\indic}[1]{\mathbbm{1}_{#1}}

\newcommand{\p}{Q}

\newcommand{\A}{\textsf{A}}
\newcommand{\OPT}{\textsf{OPT}}

\newcommand{\diff}{\textrm{d}}

\newcommand{\err}{\eta}

\newcommand{\con}{c}
\newcommand{\rob}{r}
\newcommand{\z}{\Tilde{y}}
\newcommand{\thresh}{\Phi}
\newcommand{\rthresh}{\Tilde{\Phi}(\lambda,\rho,y)}
\newcommand{\pmax}{p^*}
\newcommand{\AOS}{\A^{\operatorname{LS}}_b}
\newcommand{\AOM}{\A^{\operatorname{OM}}_\lambda}
\newcommand{\RAOM}{\A^{\operatorname{OM}}_{\lambda,\rho}}
\newcommand{\ASR}{\A^{\operatorname{SR}}_{\lambda, \rho}}

\newtheorem{theorem}{Theorem}[section]
\newtheorem{corollary}{Corollary}[theorem]
\newtheorem{lemma}[theorem]{Lemma}

\newtheorem{definition}{Definition}[section]

\newtheorem{proposition}[theorem]{Proposition}

\usepackage{mathtools}

\title{On Tradeoffs in Learning-Augmented Algorithms}

\author{%
Ziyad Benomar \\ CREST, Ecole polytechnique,\\
ENSAE, Fairplay joint team, Palaiseau\\ \\
\and
Vianney Perchet \\ CREST, ENSAE, Criteo AI LAB \\ Fairplay joint team, Paris
}

\date{}

\begin{document}

\maketitle

\begin{abstract}
The field of learning-augmented algorithms has gained significant attention in recent years. These algorithms, using potentially inaccurate predictions, must exhibit three key properties: consistency, robustness, and smoothness. In scenarios where distributional information about predictions is available, a strong expected performance is required. Typically, the design of these algorithms involves a natural tradeoff between consistency and robustness, and previous works aimed to achieve Pareto-optimal tradeoffs for specific problems. However, in some settings, this comes at the expense of smoothness. In this paper, we demonstrate that certain problems involve multiple tradeoffs between consistency, robustness, smoothness, and average performance.
\end{abstract}

\section{Introduction}

Many decision-making problems under uncertainty are commonly studied using competitive analysis. In this context, the performance of online algorithms, operating under uncertainty, is compared to that of the optimal offline algorithm, which has full knowledge of the problem instance. 
While competitive analysis provides a rigorous method for evaluating online algorithms, it is often overly pessimistic. In real-world scenarios, decision-makers can have some prior knowledge, though possibly imperfect, about the complete problem instance. For example, predictions of unknown variables might be obtained via machine learning models, or an expert might provide advice on the best course of action. This more realistic setting was formalized by \citet{lykouris2018competitive} and \citet{purohit2018improving} leading to the development of what is now known as learning-augmented algorithms.
In this paradigm, the algorithm receives predictions about the current problem instance, but without any guarantees on their accuracy, and must satisfy three main properties:
\begin{itemize}
    \item Consistency: perform almost as well as the optimal offline algorithm if the predictions are perfect.
    \item Robustness: maintain a performance level close to the worst-case scenario without predictions when the predictions are arbitrarily bad.
    \item Smoothness: the performance should degrade gracefully as the prediction error increases.
\end{itemize}

\paragraph{Consistency, robustness, and brittleness.}
Consider a minimization problem under uncertainty, and let $\ALG$ be an algorithm augmented with a prediction $y$ of an unknown parameter $x$. The input of the algorithm can contain parameters other than $x$, but for simplicity, we denote by $\ALG(x,y)$ the value of the objective function achieved by $\ALG$, and $\OPT(x)$ the value of the optimal offline algorithm. The consistency $\con$ and robustness $\rob$ of $\ALG$ are defined as
\[
\con = \sup_{x}\frac{\ALG(x,x)}{\OPT(x)}\,
\qquad \text{and} \qquad
\rob = \sup_{x,y} \frac{\ALG(x,y)}{\OPT(x)}\;.
\]
Consistency is the worst-case ratio when the prediction is perfectly accurate, i.e. $y=x$, while robustness is the worst-case ratio with adversarial prediction.
Most research on learning-augmented algorithms focuses on achieving good tradeoffs between consistency and robustness. Some studies also establish algorithms with Pareto-optimal tradeoffs, i.e. no algorithm can have simultaneously better consistency and better robustness. However, the proposed algorithms in such studies sometimes lack smoothness. Specifically, their worst-case performance can degrade abruptly, moving from the consistency bound when the predictions are perfect to the robustness bound even with an arbitrarily small error in the prediction. Following the terminology of \citep{angelopoulos2024overcoming},  we say that such an algorithm is brittle.
\begin{definition}[Brittleness]\label{def:brittle}
an algorithm $\ALG$ with robustness $\rob$ is brittle if
\[
\forall \varepsilon >0: \sup_{x,y:\; |x-y| \leq \varepsilon x} \frac{\ALG(x,y)}{\OPT(x)} = \rob\;.
\]
\end{definition}
In real-world scenarios, predictions are rarely perfect. As a result, the only reliable guarantee for brittle algorithms is the robustness bound, which is at best equivalent to the worst-case bound without predictions. This greatly limits the practical usefulness of these algorithms. In the case of the \textit{one-way trading} problem, \cite{angelopoulos2024overcoming} demonstrated in a very recent work that any algorithm achieving a Pareto-optimal tradeoff between consistency and robustness is brittle. This finding implies that, in some problems, achieving smoothness requires deviating from the Pareto-optimal tradeoff between consistency and robustness.

\paragraph{Average-case performance.}
In the context of learning-augmented algorithms, the consistency, robustness, and smoothness of an algorithm represent worst-case guarantees with respect to the prediction error. On the other hand, there are numerous scenarios where the decision-maker might know some information about the distribution of the prediction, which motivates the design of algorithms with good average performance \citep{dutting2021secretaries, gupta2022augmenting, benomar2023advice, henzinger2023complexity, cohen2024max}. Nevertheless, while achieving a balance between worst-case and average-case performance has been widely studied in various fields of algorithm design and machine learning \citep{szirmay1998worst, witt2005worst, peikert2007lattices, antunes2009worst, chuangpishit2018average, rice2021robustness, robey2022probabilistically}, this aspect has not yet been investigated in the context of learning-augmented algorithms for the prediction error.

\subsection{Contributions}
In this work, we explore various tradeoffs that arise in the design and analysis of learning-augmented algorithms. While existing literature has primarily focused on the tradeoff between consistency and robustness, our investigation centers on the tradeoffs between consistency and smoothness, as well as the relationships between the standard criteria for learning-augmented algorithms—namely consistency, robustness, and smoothness—and their average performance under stochastic assumptions regarding predictions.

We begin by examining the \textit{line search problem}, revisiting the algorithm proposed by \citet{angelopoulos2019online}. This algorithm achieves a Pareto-optimal tradeoff between consistency and robustness among deterministic algorithms, but we demonstrate that it is inherently brittle. We show that this brittleness can be mitigated by introducing randomness into the predictions used by the algorithm. The variance of this randomization is quantified by a parameter $\rho \geq 0$. Our analysis reveals that tuning this parameter leads to opposing effects on the consistency and smoothness of the algorithm, thus yielding a tradeoff between these two criteria.

Next, we apply a similar approach to the \textit{one-max search problem}. We examine the Pareto-optimal algorithm introduced by \citet{sun2021pareto}, and we show its brittleness. Furthermore, we demonstrate how randomization can be used to guarantee smoothness at the cost of consistency. Once again, the resulting tradeoff is governed by a parameter $\rho \geq 0$.

Finally, we address the \textit{ski-rental problem}, proposing an algorithm that generalizes that of \citet{purohit2018improving}. Through a tight analysis of its performance, we prove that the Pareto-optimal tradeoff between consistency and robustness can be achieved with different levels of smoothness. However, we show that striving for optimal smoothness degrades the average-case performance of the algorithm, assuming that the prediction induces the correct decision (renting or buying at time $0$) with a probability $\p \in [\tfrac{1}{2},1]$. In this context, a parameter $\rho \in [0,1]$ can be utilized to tune the levels of smoothness and average-case performance, all while maintaining fixed levels of consistency and robustness.

Additionally, we conduct numerical experiments for the three problems studied in the paper, highlighting the various tradeoffs demonstrated in our analysis.

\subsection{Related work}

\paragraph{Learning-augmented algorithms.}
The design of learning-augmented algorithms \citep{lykouris2018competitive,purohit2018improving} relies on using machine-learned advice to go beyond worst-case limitations. These algorithms operate under the assumption that the decision-maker has access to noisy predictions about certain problem parameters. The goal of learning-augmented algorithms is to improve performance if the predictions are accurate, while also ensuring robustness in the face of incorrect or adversarial predictions. Many fundamental algorithmic problems were studied in this setting, such as ski rental \citep{gollapudi2019online, diakonikolas2021learning, antoniadis2021learning, shin2023improved}, caching \citep{lykouris2018competitive, chlkedowski2021robust,antoniadis2023online, antoniadis2023paging}, scheduling \citep{purohit2018improving, merlis2023preemption, lassota2023minimalistic, benomarnon}, and the design of data structures \citep{kraska2018case,lin2022learning, zeynalirobust, benomar2024learning}.

\paragraph{Overcoming brittleness.}
Pareto-optimal tradeoffs between consistency and robustness were studied in \citep{angelopoulos2019online, bamas2020primal, wei2020optimal, sun2021pareto, angelopoulos2023online}. However, the proposed algorithms do not always have smoothness guarantees. For example, \cite{angelopoulos2024overcoming} proved that any Pareto-optimal algorithm for the one-way trading problem is necessarily brittle (Definition \ref{def:brittle}). In this paper, we will show in the line search and in the one-max search problems how a randomized deviation from the Pareto-optimal algorithm allows for overcoming brittleness. A similar approach was used to guarantee smoothness in non-clairvoyant scheduling with limited predictions \citep{benomarnon}.

\paragraph{Line search.} The line search problem \citep{beck1964linear}, also known as the cow path problem, consists of finding a hidden target on an infinite line starting from an initial position, without any information regarding the direction or distance to the target. The goal is to minimize the total distance traveled before reaching the target. The best deterministic algorithm is based on doubling the search distance in alternating directions, and it ensures a competitive ratio of $9$ \citep{beck1970yet, baezayates1993searching}. The line search problem has been extensively studied in the learning-augmented framework with different types of predictions \citep{angelopoulos2023online, angelopoulos2024search}.

\paragraph{One-max search.} In the one-max search problem \citep{el2001optimal}, the decision-maker observes a sequence of adversarially chosen prices $p_1,\ldots,p_n \in [L,U]$, with $0<L<U$. At each step $i$, the price $p_i$ is revealed to the decision-maker, and the latter can decide to stop the game and have a payoff of $p_i$, or reject it irrevocably and move to the next observation. The best deterministic algorithm for this problem consists simply in selecting the first price larger than $\sqrt{LU}$, which guarantees a payoff of at least $\sqrt{L/U} \max_{i \in [n]} p_i$. This problem, as well as its randomized version—\textit{online conversion}—were studied in the learning-augmented setting with a prediction of the maximal price \citep{sun2021pareto}.

\paragraph{Ski-rental}
In the ski-rental problem, the decision maker faces a daily choice between renting a ski at a unit cost or buying it for a one-time cost of $b$, after which skiing becomes free. The length of the ski season, $x$, is unknown, and the goal is to minimize the total cost of renting and buying. A straightforward algorithm for this problem is renting for the first $b-1$ days and then buying at day $b$, resulting in a competitive ratio of 2, which is the best achievable by any deterministic algorithm \citet{karlin1988competitive}. The best competitive ratio with randomized algorithms is $\frac{e}{e-1}$ \citep{karlin1994competitive}.

\section{Smooth Algorithm for Line Search}\label{sec:line-search}
In the line search problem, a target is hidden in an unknown position $x\in \R$ on the line, with $|x| \geq 1$, and the searcher, initially at the origin of the line $O$, must find the target, minimizing the total traveled distance. The optimal offline algorithm only travels a distance of $|x|$ to find the target. On the other hand, the searcher, ignoring if $x$ is to the left or the right side of $O$, must alternate the search direction multiple times before finding the target. Any deterministic algorithm for this problem can be defined as an iterative strategy, parameterized by an initial search direction $s_0 \in \{-1,1\}$ and a sequence of turn points $(d_i)_{i \in \N} \in [1,\infty]^\N$. At the beginning of any iteration $i \geq 0$, the searcher is located at the origin $O$, then it travels a distance $d_i$ in the direction $(-1)^{i} s_0$ and returns to the origin. The algorithm terminates when the position $x$ is reached.

\paragraph{Pareto-optimal algorithm.}
Given a prediction $y$ of $x$, \cite{angelopoulos2023online} designed an algorithm $\AOS$ that has a consistency of $\frac{b+1}{b-1}$, and robustness of $1+\frac{2b^2}{b-1}$, where $b \geq 2$ is a hyperparameter of the algorithm. Denoting by $\AOS(x,y)$ the distance traveled by $\AOS$ to find the target $x$ given the prediction $y$, the consistency and robustness guarantees can be written as
\begin{align*}
\forall x:\quad &\frac{\AOS(x,x)}{|x|} \leq \frac{b+1}{b-1}\;,\\
\forall x,y:\quad &\frac{\AOS(x,y)}{|x|} \leq 1 + \frac{2b^2}{b-1} \;.
\end{align*}

Moreover, the author proves that these consistency and robustness levels are Pareto-optimal. The proposed algorithm has a simple structure: let $k_y \in \N$ such that $b^{k_y-2} < |y| \leq b^{k_y}$, and $\gamma_y = b^{k_y} / |y| \geq 1$, the algorithm $\AOS$ is defined by the initial search direction $s_0 = (-1)^{k_y}\text{sign}(y)$ and the turn points $d_i = b^i/\gamma_y$ for all $i \geq 0$. The algorithm is defined so that, during the iteration $k_y$, the searcher travels a distance of $|y|$ in the direction given by $\text{sign}(y)$, i.e. it reaches the position $y$ exactly at the turning point of iteration $k_y$.

\subsection{Brittleness of the Pareto-optimal algorithm} 
In the following, we will prove that $\AOS$ is brittle, in the sense of Definition \ref{def:brittle}, then we will demonstrate how a simple randomization idea enables making the algorithm smooth.
To better understand the impact of the prediction error on the performance of $\AOS$, we first prove an expression of $\AOS(x,y)/x$ as a function of $y$ and $x$.
\begin{lemma}\label{lem:wc-po-os} 
Let $x\geq 1$, $y>0$ and $j = \lceil \frac{\ln(x/y)}{2 \ln b} \rceil \in \mathbb{Z}$, so that $1 \leq \frac{y}{x} b^{2j} < b^2$. It holds that
\begin{align*}
\frac{\AOS(x,y)}{x} 
= 1 + \frac{2b^2}{b-1} \cdot b^{2(j-1)}\frac{y}{x} - o(1/x)\;.
\end{align*}
\end{lemma}

\begin{figure}[h!]
    \centering
    \includegraphics[width=0.6\linewidth]{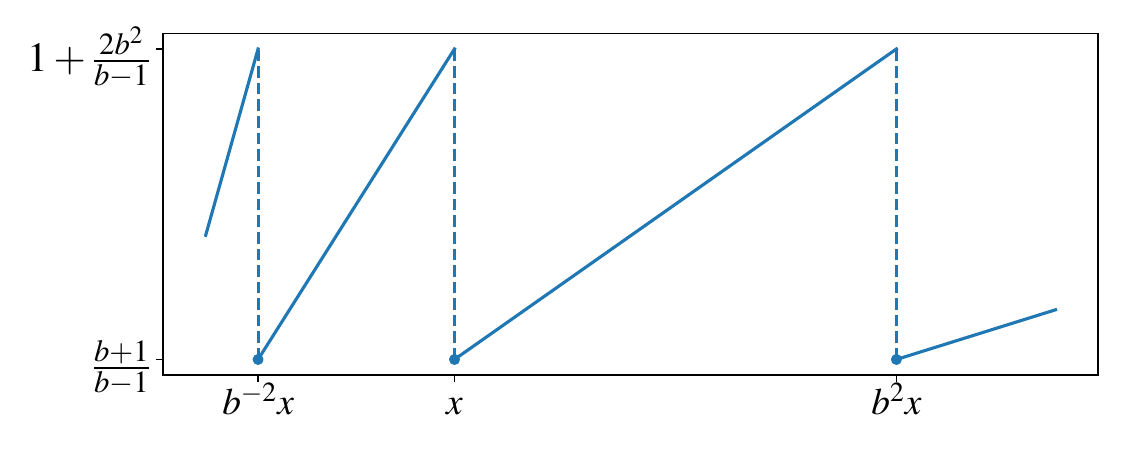}
    \caption{The mapping $y \mapsto \AOS(x,y)/x$ for $x$ arbitrary large and $y \in [\tfrac{x}{b^2}, b^2x]$.}.
    \label{fig:lb-onlineSearch}
\end{figure}

\begin{proof}
By definition of $j$, we have $1 \leq b^{2j} y/x < b^2$, and recalling that $y = b^{k_y}/\gamma_y$, we obtain
\[
\frac{b^{k_y+2j-2}}{\gamma_y} < x \leq \frac{b^{k_y + 2j}}{\gamma_y}\;.
\]
Using these inequalities, we deduce that the algorithm finds the target at iteration $k_y+2j$. Indeed, as $y>0$, at iteration $k_y+2j-2$, the searcher travels along the positive branch until reaching the position $b^{k_y+2j-2}/\gamma_y$ and returns to the origin, then it explores the negative branch at iteration $k_y+2j-1$. The total distance traveled up to that point is $2\sum_{i=0}^{k_y+2j-1} b^i/\gamma_y$. At iteration $k_y + 2j$, since $x \leq b^{k_y + 2j}/\gamma_y$, the algorithm finds the target and terminates. The total traveled distance is therefore
\begin{align}
\AOS(x,y)
&= x + 2\sum_{i=0}^{k_y+2j-1} \frac{b^i}{\gamma_y} \nonumber\\
&= x + \frac{2}{b-1}\left( \frac{b^{k_y+2j}}{\gamma_y} - \frac{1}{\gamma_y} \right) \nonumber\\
&= x + \frac{2b^2}{b-1}\left( b^{2(j-1)} y - \frac{1}{\gamma_y b^2} \right)\;, \nonumber
\end{align}
where we used in the last inequality that $y = b^{k_y}/\gamma_y$. Finally, using that $\gamma_y \geq 1$, we deduce that
\[
1 + \frac{2b^2}{b-1} \left(b^{2(j-1)}\frac{y}{x} - \frac{1}{b^2 x}\right)
\leq \frac{\AOS(x,y)}{x} 
\leq 1 + \frac{2b^2}{b-1} \cdot b^{2(j-1)}\frac{y}{x}\;.
\]
\end{proof}

The expression proved in Lemma \ref{lem:wc-po-os} is illustrated in Figure \ref{fig:lb-onlineSearch} for $x$ arbitrary large and  $y \in [x/b^2, b^2x]$. It shows that the ratio $\AOS(x,y)/x$ increases smoothly from the consistency to the robustness bound if the prediction $y$ is larger than $x$, but presents a discontinuity, going immediately from the consistency to robustness bound if $y<x$ and arbitrarily close to $x$. This proves the brittleness of $\AOS$, which can be formally stated as follows.
\begin{proposition}[$\AOS$ is brittle]\label{prop:AOS-brittle}
For any $\varepsilon >0$, it holds that
\[
\sup_{x,y: |x-y|\leq \varepsilon x} \frac{\AOS(x,y)}{|x|} = 1+\frac{2b^2}{b-1} \;.
\]
\end{proposition}

\begin{proof}
Let $\varepsilon > 0$, $x \geq 1$,  and $y \in [\max\{\frac{x}{b^2}, (1-\varepsilon)x\}, x)$. Since $\frac{x}{b^2} \leq y < x$, the variable $j$ from Lemma \ref{lem:wc-po-os} equals $1$, which yields
\[
\frac{\AOS(x,y)}{x} 
\geq 1 + \frac{2b^2}{b-1} \left(\frac{y}{x} - \frac{1}{b^2 x}\right)\;,
\]
and taking $y$ arbitrarily close to $x$ gives
\[
\sup_{y: |x-y|\leq \varepsilon x} \frac{\AOS(x,y)}{x} \geq 1+\frac{2b^2}{b-1}\left(1 - \frac{1}{b^2x}\right)
= \left( 1+\frac{2b^2}{b-1} \right) - \frac{2}{(b-1)x}\;.
\]
Finally, given that $1+\frac{2b^2}{b-1}$ is also an upper bound on $\AOS(x,y)$, taking arbitrarily large concludes the proof.
\end{proof}

\subsection{Smoothness via randomization}
In all the following, we assume without loss of generality that $x > 0$. The worst-case ratio $\AOS(x,y)/x$, given in Proposition \ref{prop:AOS-brittle}, occurs when $x = y+\varepsilon$ with $\varepsilon$ arbitrarily small. To avoid it, we perturb $y$ and run instead the algorithm $\AOS$ with a randomized prediction of the form $\z = (1+\rho \xi)y$, where $\rho > 0$ is a hyperparameter and $\xi$ a positive random variable. 

\begin{theorem}\label{thm:smooth-LS}
Let $b\geq 2$, $\rho \in [0,1]$, and $\xi$ a random variable with tail distribution $\Pr(\xi \geq t) = \frac{1}{(1+t)^2}$ for all $t \geq 0$. Then for any $x \geq 1$ and $y \in \R$, denoting by $\err = |x-y|$, we have for $\z = (1+\rho\xi) y$ that
\[
\frac{\E_\xi[\AOS(x,\z)]}{x}
\leq \frac{b+1+2\rho}{b-1} + \left\{
 \begin{array}{ll}
    \dfrac{2(1+\rho)}{b-1} \cdot \dfrac{\err}{x} & \text{if } y \geq x \\[15pt]
    \dfrac{4(b+1)}{\rho} \cdot \dfrac{\err}{x} & \text{if } y<x
\end{array}
\right.
\]
and we have with probability $1$ that
\[
\frac{\AOS(x, \z)}{x} \leq 1 + \frac{2b^2}{b-1}\;.
\]
\end{theorem}

The first bound in the previous theorem establishes the algorithm's consistency and smoothness, while the second bound characterizes its robustness, which remains unaffected by randomizing the prediction. Beyond the consistency-robustness tradeoff governed by the parameter $b$, the algorithm also exhibits a tradeoff between consistency and smoothness, governed by the parameter $\rho \in [0,1]$. For $\rho = 0$, the algorithm is identical to that of \cite{angelopoulos2019online} and has optimal levels of consistency and smoothness. For $\rho > 0$, the smoothness factor for $y<x$ improves, but the algorithm becomes less consistent. Although the algorithm is well-defined for all values of $\rho > 0$, we limit our analysis to $\rho \in [0,1]$, as this range allows for simpler expressions of the upper bound. Let us now prove the theorem

\begin{proof}
The online algorithm $\AOS$ has a robustness of $1+\frac{2b^2}{b-1}$. This guarantee remains true with any arbitrary prediction, in particular with $\z$, which gives almost surely that
\[
\frac{\AOS(x, \z)}{x} \leq 1 + \frac{2b^2}{b-1}\;.
\]
Regarding the consistency and smoothness of the algorithm, we give separate proofs depending on the position of $y$ relative to $x$.

If $y \geq x$, then $y = x + \err$, the perturbed prediction $\z$ is larger than $x$ almost surely, and $j = \lceil \tfrac{\ln(x/\z)}{2\ln(b)} \rceil \leq 0$. Thus, Lemma \ref{lem:wc-po-os} with $x$ and $\z$ yields
\begin{align*}
\frac{\AOS(x,\z)}{x}
&\leq 1 + \frac{2}{b-1} \cdot \frac{\z}{x}\\
&= 1 + \frac{2 (1+\rho \xi)}{b-1} \cdot \left(1 + \frac{\err}{x}\right)\\
&= \frac{b-1+2\rho \xi}{b-1} + \frac{2(1+\rho \xi)}{b-1} \cdot \frac{\err}{x}\;
\end{align*}
and using that $\E[\xi] = 1$ gives in expectation
\[
\frac{\E_\xi[\AOS(x,\z)]}{x}
\leq \frac{b-1+2\rho}{b-1} + \frac{2(1+\rho)}{b-1} \cdot \frac{\err}{x}\;.
\]

On the other hand, if $y < x$, let first prove that the upper bound of the theorem exceeds the robustness bound when $y \leq x/b^2$, hence it is true almost surely in that case.
For $y \leq x/b^2$, the absolute prediction error satisfies $\err/x = 1 - y/x \geq 1-1/b^2$, and given that $\rho \leq 1$, we obtain 
\begin{align*}
\frac{b+1+2\rho}{b-1} + \frac{4(b+1)}{\rho} \cdot \frac{\err}{x} 
&\geq \frac{b+1}{b-1} + 4 (b+1)\left(1 - \frac{1}{b^2}\right)\\
&= 1 + \frac{2b^2}{b-1}\left( \frac{1}{b^2} + 2\left(1-\frac{1}{b^2}\right)^2 \right)\\
&= 1 + \frac{2b^2}{b-1}\left( 1 + \left( 1 - \frac{3}{b^2} + \frac{2}{b^4} \right) \right)\\
&\geq 1 + \frac{2b^2}{b-1}\\
&\geq \frac{\AOS(x,\z)}{x} \qquad (\text{w.p. } 1)\;.
\end{align*}
Consequently, in the rest of the proof, we focus on showing the claimed result when $y \in (\frac{x}{b^2}, x)$. In that case, the random variable $\z$ takes values in $(x/b^2, \infty)$. By Lemma \ref{lem:wc-po-os}, using that $\z = (1+\rho \xi) y$, we obtain that
\begin{align*}
\frac{\AOS(x,\z)}{x}
&\leq 1 + \frac{2 \z/x}{b-1} \times \left\{
 \begin{array}{ll}
    b^2 & \text{if } \z \in (x/b^2,x) \\[10pt]
    1 & \text{if } \z \geq x
\end{array}
\right.\\
&= 1 + \frac{2 y/x}{b-1}\left( (1+\rho \xi) \indic{\xi \geq \frac{1}{\rho}(\frac{x}{y}-1)} + b^2 (1+\rho \xi)\indic{\xi < \frac{1}{\rho}(\frac{x}{y}-1)} \right)\\
&= 1 + \frac{2 y/x}{b-1}\left( 1+\rho \xi + (b^2-1) (1+\rho \xi)\indic{\xi < \frac{1}{\rho}(\frac{x}{y}-1)} \right)\;,\\
\end{align*}
which gives in expectation
\begin{equation}\label{eq:os-generic-xi-ub}
\frac{\E[\AOS(x,\z)]}{x}
\leq 1 + \frac{2 y/x}{b-1}\left( 1+\rho + (b^2-1) \big( \Pr( \xi < \tfrac{1}{\rho}(\tfrac{x}{y}-1)) + \rho \E\big[\xi \indic{\xi < \frac{1}{\rho}(\frac{x}{y}-1)}\big] \big) \right)\;. 
\end{equation}
For all $s \geq 0$, it holds that $\Pr(\xi < s) = 1-\frac{1}{(1+s)^2} = \frac{s^2 + 2s}{(1+s)^2}$, and
\begin{align*}
\E[\xi \indic{\xi < s}]
&= \int_0^\infty \Pr(\xi \indic{\xi < s} \geq u) \diff u\\
&= \int_0^s \Pr(\xi \in [u,s)) \diff u\\
&= \int_0^s \Pr(\xi \geq u) \diff u - \int_0^s \Pr(\xi \geq s) \diff u\\
&= \int_0^s \frac{\diff u}{(1+u)^2} - \frac{s}{(1+s)^2}\\
&= 1 - \frac{1}{1+s} - \frac{s}{(1+s)^2}\\
&= \frac{s^2}{(1+s)^2}\;.
\end{align*}
Therefore, using that $\rho \leq 1$ then $s>0$, we obtain
\begin{align*}
\Pr( \xi < s) + \rho \E\big[\xi \indic{\xi < s}\big]
&= \frac{2s+(1+\rho)s^2}{(1+s)^2}
\leq \frac{2s+2s^2}{(1+s)^2}
= \frac{2s}{1+s} 
\leq 2s\;.
\end{align*}
Substituting into \eqref{eq:os-generic-xi-ub} with $s = \frac{1}{\rho}(\frac{x}{y}-1)$ yields
\begin{align*}
\frac{\E[\AOS(x,\z)]}{x}
&\leq 1 + \frac{2y/x}{b-1} \left( 1+\rho + \frac{2(b^2-1)}{\rho}\Big(\frac{x}{y}-1\Big)  \right)\\
&= \frac{b-1+2(1+\rho)y/x}{b-1} + \frac{4(b+1)}{\rho} \cdot \frac{x-y}{x}\\
&\leq \frac{b+1+2\rho}{b-1} + \frac{4(b+1)}{\rho} \cdot \frac{\err}{x}\;,
\end{align*}
which concludes the proof.
\end{proof}

\section{Smooth Algorithm for One-Max Search}
In the one-max search problem, a decision-maker sequentially observes prices $p_1, \ldots, p_n \in [L, U]$, where $0 < L < U$, and upon observing each price $p_i$  they must decide either to select it, halting the process and receiving a payoff of $p_i$, or to reject it irrevocably and move on to the next price. 

Let $\ALG$ denote an online algorithm for this problem. We use $\ALG(p)$ to denote the price selected by $\ALG$ when given an input instance  $p = (p_1, \ldots, p_n)$, and we use $\pmax = \max_{i \in [n]} p_i$ to denote the highest price in the instance. Since this is a maximization problem The competitive ratio of 
$\ALG$ is defined as:
\begin{equation}\label{eq:CR-1M}
\CR(\ALG) = \inf_{p} \frac{\ALG(p)}{\pmax}
\end{equation}
where the infimum is taken over all possible price sequences of arbitrary length, with prices in the range $[L, U]$.
In some works, the competitive ratio is alternatively defined as the inverse of our definition, i.e., 
$
\CR(\ALG) = \sup_{p} \frac{\pmax}{\ALG(p)}.
$
However, since our objective is to prove smoothness guarantees, it is more convenient to define the competitive ratio as in \eqref{eq:CR-1M}.
Finally, let  $\theta = U/L$. 
Without loss of generality, we can assume that the prices are in the interval 
$[1, \theta]$.

\paragraph{Pareto-Optimal Algorithm.} Given a prediction $y$ of the maximum price, the consistency $\con$ and robustness $\rob$ of any algorithm $\ALG$ are defined as
$\con = \inf_p \frac{\ALG(p, \pmax)}{\pmax}$
and 
$\rob = \inf_{p,y} \frac{\ALG(p,y)}{\pmax}$.
In this setting, \citet{sun2021pareto} proved for all $\lambda \in [0,1]$ an algorithm $\AOM$ with consistency $\con(\lambda)$ and robustness $\rob(\lambda)$, defined as the unique solution of the system
\begin{equation}\label{eq:con-rob-relation-OM}
\frac{1}{\con(\lambda)} = \theta \cdot \rob(\lambda)
\quad \text{and} \quad
\frac{1}{\con(\lambda)} = \frac{\lambda}{\rob(\lambda)} + 1-\lambda\;.
\end{equation}
The proposed algorithm is a threshold policy, where the first price at least equal to the threshold $\thresh(\lambda,y)$ is selected, with
\[
\thresh(\lambda,y) = \left\{
 \begin{array}{lll}
    \tfrac{1}{\con(\lambda)} & \text{if } y \in [1,\frac{1}{\con(\lambda)}) \\[12pt]
     \tfrac{\lambda}{\rob(\lambda)} + (1-\lambda) \con(\lambda) y & \text{if } y \in [\frac{1}{\con(\lambda)}, \frac{1}{\rob(\lambda)})\\[12pt]
     \tfrac{1}{\rob(\lambda)} & \text{if } y \in [\frac{1}{\rob(\lambda)}, \theta]
\end{array}
\right. .
\]
If all the prices are less than $\thresh(\lambda,y)$, then the algorithm selects the last price $p_{n}$. Furthermore, the authors prove that the levels of consistency and smoothness of $\AOM$ are Pareto-optimal. However, they do not provide any smoothness guarantees.

\subsection{Brittleness of the Pareto-optimal algorithm}
We will prove in the following that this algorithm is brittle, and we will show, similarly to the line search problem, how this brittleness can be overcome via randomization. In the following, we simply write $\con, \rob$ instead of $\con(\lambda), \rob(\lambda)$, and we denote by $\err := |\pmax - y|$ the prediction error.
Our first result is that the competitive ratio of $\A_\lambda$ degrades smoothly as a function of the prediction error when $y \in [1,\frac{1}{\rob})$. Then, we will prove the brittleness of the algorithm in Proposition \ref{prop:OM-brittle} by considering $y \geq \frac{1}{\rob}$. 
The following lemma shows smoothness for $y \in [1, \tfrac{1}{\con})$.

\begin{lemma}\label{lem:OM-case1-smooth}
If $y \in [1, \frac{1}{\con})$, then 
\[
\frac{\AOM(p,y)}{\pmax} 
\geq  \con  - \frac{\con \cdot \err}{\pmax}\;.
\]
\end{lemma}

\begin{proof}
If $y \in [1, \frac{1}{\con})$, then $\thresh(\lambda, y) = 1/\con$. If $\pmax < 1/\con$ then all the observed prices are below the threshold, and the algorithm selects $p_{n}$, which is $1$ in the worst case, hence
\[
\frac{\AOM(p,y)}{\pmax} \geq \frac{1}{\pmax} \geq \con\;.
\]
On the other hand, if $\pmax \geq 1/\con$, then the value selected by the algorithm is at least $1/\con$ and
\[
\frac{\AOM(p,y)}{\pmax} 
\geq \frac{1/\con}{\pmax} \geq \frac{y}{\pmax}
= 1 - \frac{\pmax-y}{\pmax}
= 1 - \frac{\err}{\pmax}\;.
\]
We deduce from both cases that
\[
\frac{\AOM(p,y)}{\pmax} 
\geq \min \Big( \con \;, \; 1 - \frac{\err}{\pmax} \Big)
\geq \con  - \frac{\con \cdot \err}{\pmax}\;,
\]
where the last inequality holds because each term is at most $1$.
\end{proof}

The next lemma proves a similar result for $y \in [\tfrac{1}{\con}, \tfrac{1}{\rob})$.

\begin{lemma}\label{lem:OM-case2-smooth}
If $y \in [\frac{1}{\con}, \frac{1}{\rob})$, then 
\[
\frac{\AOM(p,y)}{\pmax} 
\geq  \con  - (1-\lambda) \max(1, \tfrac{\con}{\lambda})   \frac{\con \cdot \err}{\pmax}\;.
\]
\end{lemma}

\begin{proof}
Assume that $y \in [1/\con, 1/\rob)$, then the acceptance threshold is $\thresh(\lambda, y) = \lambda/\rob + (1-\lambda)\con y$. If $\pmax \geq \thresh(\lambda,y)$ then the price selected by the algorithm is at least $\thresh(\lambda,y)$, which gives that
\begin{align*}
\frac{\AOM(p,y)}{\pmax} 
\geq \frac{\thresh(\lambda,y)}{\pmax}
= \frac{\lambda}{\rob \pmax} + (1-\lambda)\con \frac{y}{\pmax}\;.
\end{align*}
We have from \eqref{eq:con-rob-relation-OM} that $\pmax \leq \theta = 1/(\con \rob)$, hence $\frac{\lambda}{\rob \pmax} \geq \lambda \con$. Additionally, $y \geq \pmax - \err$, which gives
\begin{align}
\frac{\AOM(p,y)}{\pmax} 
&\geq  \lambda \con + (1-\lambda)\con \Big(1 - \frac{\err}{\pmax}\Big) \nonumber\\
&= \con - (1-\lambda) \frac{\con \cdot \err}{\pmax}\;. \label{eq:OM-case2-lb1}
\end{align}
On the other hand, if $\pmax < \thresh(\lambda, y)$, then the value selected by the algorithm can be as low as $1$ in the worst case, thus 
\begin{equation}\label{eq:OM-case2-p<T}
\frac{\AOM(p,y)}{\pmax} 
\geq \frac{1}{\pmax}
= \frac{1}{\thresh(\lambda,y)} \cdot \frac{\thresh(\lambda,y)}{\pmax}\;. 
\end{equation}
Using \eqref{eq:con-rob-relation-OM}, we can write
\begin{align}
\thresh(\lambda, y)
&= \frac{\lambda}{\rob} + (1-\lambda)\con y \nonumber\\
&= \frac{1}{\con} - (1-\lambda) + (1-\lambda)\con y \nonumber\\
&= \frac{1}{\con} + (1-\lambda) \con \Big( y - \frac{1}{\con} \Big)\;, \label{aligneq:thresh-expression}
\end{align}
hence
\[
y - \thresh(\lambda,y) = \big( 1 - (1-\lambda)c \big) \Big( y - \frac{1}{c} \Big)\;,
\]
and it follows from $\pmax < \thresh(\lambda,y)$ that
\[
y - \frac{1}{c}
= \frac{y - \thresh(\lambda,y)}{1 - (1-\lambda)c}
\leq \frac{y - \pmax}{1 - (1-\lambda)c}
= \frac{\err}{1 - (1-\lambda)c}\;.
\]
Substituting into \eqref{aligneq:thresh-expression} then using \eqref{eq:con-rob-relation-OM} yields
\begin{align*}
\thresh(\lambda, y)
&\leq \frac{1}{\con} + \frac{(1-\lambda)\con}{1 - (1-\lambda) \con}\cdot \err\\
&= \frac{1}{\con} + \frac{(1-\lambda)}{\frac{1}{\con} - (1-\lambda)}\cdot \err\\
&= \frac{1}{\con} + \Big(\frac{1}{\lambda}-1\Big) \rob \cdot \err\;.
\end{align*}
Finally, we use this upper bound on $\thresh(\lambda,y)$ in \eqref{eq:OM-case2-p<T} and obtain that
\begin{align*}
\frac{\AOM(p,y)}{\pmax} 
&= \frac{1}{\thresh(\lambda,y)} \cdot \frac{\thresh(\lambda,y)}{\pmax} \\
&\geq \frac{1}{\frac{1}{\con} + (\frac{1}{\lambda}-1) \rob \cdot \err} \cdot \frac{\thresh(\lambda,y)}{\pmax}\\
&= \frac{c}{1 + (\frac{1}{\lambda}-1) \rob \con \cdot \err} \cdot \frac{\thresh(\lambda,y)}{\pmax}\\
&\geq \con \cdot \frac{\thresh(\lambda,y)}{\pmax} \left( 1 - (\tfrac{1}{\lambda}-1) \rob \con \cdot \err \right)\\
&\geq \con -  \thresh(\lambda,y) (\tfrac{1}{\lambda}-1) \rob \con^2 \cdot \frac{\err}{\pmax}\;.
\end{align*}
The penultimate line follows from the inequality $\frac{1}{1+u} \geq 1-u$, which holds for all $u\geq 0$, and the last line from the assumption $\pmax \leq \thresh(\lambda,y)$. Given that $y \leq 1/\rob$, we have
\[
\thresh(\lambda,y) = \frac{\lambda}{\rob} + (1-\lambda)\con y
\leq \frac{1}{\rob}\big( \lambda + (1-\lambda) \con \big)
\leq \frac{1}{\rob}\;,
\]
hence
\begin{equation}\label{eq:OM-case2-lb2}
\frac{\AOM(p,y)}{\pmax} 
\geq \con -  (\tfrac{1}{\lambda}-1) \con^2 \cdot \frac{\err}{\pmax}
=  \con -  (1 - \lambda) \frac{\con}{\lambda} \cdot \frac{\con \cdot \err}{\pmax}\;. 
\end{equation}
From \eqref{eq:OM-case2-lb1} and \eqref{eq:OM-case2-lb2} we deduce that, if $y \in [1/\con, 1/\rob]$ then
\begin{align*}
\frac{\AOM(p,y)}{\pmax} 
&\geq \min\left( 
\con - (1-\lambda) \frac{\con \cdot \err}{\pmax}
\; ,\; 
\con -  (1 - \lambda) \frac{\con}{\lambda} \cdot \frac{\con \cdot \err}{\pmax}\right)\\
&= \con  - (1-\lambda) \max(1, \tfrac{\con}{\lambda})   \frac{\con \cdot \err}{\pmax}\;.
\end{align*}

\end{proof}


Finally, we demonstrate the brittleness of $\AOM$ by considering $y$ greater than, but arbitrarily close to, $\tfrac{1}{\rob}$.

\begin{proposition}[$\AOM$ is brittle]\label{prop:OM-brittle}
For any $\varepsilon > 0$, it holds that
\[
\inf_{p,y: \frac{|\pmax-y|}{\pmax} \leq \varepsilon} \frac{\AOM(p,y)}{\pmax} = \rob(\lambda)\;.
\]
    
\end{proposition}

\begin{proof}
Let $\varepsilon > 0$, and consider the instance $p = (\frac{1}{\rob}-\delta,1)$ of size $2$, with $\delta \in (0,\min(\frac{1}{\rob}-1, \frac{\varepsilon}{\rob})]$, so that $p_1$ remains at least $1$ and $|p_1-1/\rob| \leq \varepsilon p_1$. Given a prediction $y = 1/\rob$ of $\pmax = p_1$, $\AOM$ sets an acceptance threshold of $1/\rob > p_1$, hence the algorithm rejects $p_1$ and accepts $p_2 = 1$ as it is the last price in the sequence $p$. It follows that 
\[
\frac{\AOM(p,y)}{\pmax} 
= \frac{1}{1/\rob - \delta}\;,
\]
hence
\[
\inf_{p,y: \frac{|\pmax-y|}{\pmax} \leq \varepsilon} \frac{\AOM(p,y)}{\pmax} 
\leq \lim_{\delta \to 0} \frac{1}{1/\rob - \delta}
= \rob.
\]
Since, by definition, $\rob$ is also a lower bound of $\frac{\AOM(p,y)}{\pmax}$, we deduce that the inequality above is equality, which concludes the proof.
\end{proof}


\subsection{Smoothness via randomization}
As we proved in Lemmas \ref{lem:OM-case1-smooth} and \ref{lem:OM-case2-smooth}, if $y \in [1,\frac{1}{\rob})$ then performance of $\AOM$ degrades smoothly with the prediction error. The brittleness of $\AOM$ in Proposition \ref{prop:OM-brittle} arises in the case where $y \in [1/\rob,\theta]$: the ratio $\AOM(p,y)/\pmax$ is larger than $\con$ for $\pmax \geq 1/\rob$, but it drops immediately to $\rob$ for $\pmax < 1/\rob$, even arbitrarily close to $1/\rob$. 

To attenuate this extreme behavior, we randomize the threshold used when $y \in [1/\rob, \theta]$. Let $\RAOM$ the algorithm accepting the first price at least equal to the random threshold $\rthresh$ defined by
\[
U \sim \mathcal{U}[0,1] , \quad
\rthresh = \left\{
 \begin{array}{ll}
    \thresh(\lambda,y) & \text{if } y \in [1,\frac{1}{\rob}) \\[12pt]
    \dfrac{e^{-\rho U}}{\rob} & \text{if } y \in [\frac{1}{\rob}, \theta]
\end{array}
\right. \;\;.
\]

If $y \in [1,\frac{1}{\rob})$, then $\RAOM$ is equivalent to $\AOM$, thus $\RAOM$ is $\rob$-robust in that case, and the consistency and smoothness guarantees from Lemmas \ref{lem:OM-case1-smooth} and \ref{lem:OM-case2-smooth} extend to $\RAOM$. Consequently, it suffices to study $\RAOM$ when $y\in [\frac{1}{\rob}, \theta]$, and we obtain the following result.

\begin{theorem}\label{thm:smooth-OM}
Let $\lambda \in [0,1]$, $\rho \geq 0$, and let $\con = \con(\lambda)$ and $\rob = \rob(\lambda)$ as defined in \eqref{eq:con-rob-relation-OM}. For any sequence of prices $p =(p_1,\ldots,p_n) \in [1,\theta]^n$ and prediction $y \in [1,\theta]$ of $\pmax := \max_{i \in [n]} p_i$, it holds that
\[
    \frac{\E_U[\RAOM(p,y)]}{\pmax} \geq \left( \frac{1-e^{-\rho}}{\rho}\right) \rob\;,
\]
and denoting by $\eta = |\pmax-y|$, the ratio $\frac{\E_U[\RAOM(p,y)]}{\pmax}$ is at least
\begin{align*}
\left\{
 \begin{array}{lll}
    \con  - \frac{\con \cdot \err}{\pmax} & \text{if } y \in [1,1/\con) \\[15pt]
    \con  - (1-\lambda) \max(1, \tfrac{\con}{\lambda})   \frac{\con \cdot \err}{\pmax} & \text{if } y \in [1/\con, 1/\rob)\\[15pt]
    \Big(\frac{1-e^{-\rho}}{\rho} \Big)\con - \Big(\frac{\con - \rob}{\rho}\Big) \frac{\err}{\pmax} &\text{if } y \in [1/\rob, \theta] 
\end{array}
\right. \;\;,
\end{align*}

\end{theorem}

The first lower bound, independent of the prediction error $\eta$, is the robustness of the algorithm, while the second bound characterizes its consistency and smoothness.
The theorem shows that, in order to guarantee a certain level of smoothness, $\RAOM$ degrades both the consistency and robustness of $\AOM$ by a factor of $(1-e^{-\rho})/\rho$, hence exhibiting a tradeoff between smoothness and both consistency and robustness.

The consistency/smoothness bounds for $y \in [1,1/\rob)$ are proved in Lemmas \ref{lem:OM-case1-smooth} and \ref{lem:OM-case2-smooth}, and the robustness in that case is $\rob \geq \left( \tfrac{1-e^{-\rho}}{\rho}\right) \rob$ because $\RAOM$ is identical to $\AOM$. Therefore, it only remains to prove the claimed bounds for $y \in [1/\rob, \theta)$.
We demonstrate in Lemma \ref{lem:RAOM-con} the consistency and smoothness of the algorithm, while the robustness is proved in Lemma \ref{lem:RAOM-rob}.

\begin{lemma}[Consitency-Smoothness]\label{lem:RAOM-con}
For any sequence of prices $p$, if $y \in [\frac{1}{\rob}, \theta]$, then 
\[
\frac{\E_U[\RAOM(p,y)]}{\pmax} \geq 
\Big(\frac{1-e^{-\rho}}{\rho} \Big)\con - \Big(\frac{\con - \rob}{\rho}\Big) \frac{\err}{\pmax}\;.
\]
\end{lemma}
\begin{proof}
Let $y \in [\frac{1}{\rob}, \theta]$, hence $\rthresh = e^{-\rho U}/\rob$, where $U$ is a uniform random variable in $[0,1]$. If $\pmax \geq \rthresh$, then the algorithm has a reward of at least $\rthresh$, and by \eqref{eq:con-rob-relation-OM} we obtain
\[
\frac{\RAOM(p,y)}{\pmax} \geq \frac{\rthresh}{\pmax}
\geq \frac{e^{-\rho U}}{\theta \rob}
= e^{-\rho U} \con \;,
\]
and if $\pmax < \rthresh$ then
\[
\frac{\RAOM(p,y)}{\pmax} 
\geq \frac{1}{\rthresh}
= e^{\rho U} \rob\;.
\]
Let us denote by $s = -\ln(\rob \pmax)$. Observing that
\[
\pmax \geq \rthresh
\iff \pmax \geq e^{-\rho U}/\rob
\iff e^{\rho U} \geq \frac{1}{\rob \pmax}
\iff U \geq \frac{-\ln(r \pmax)}{\rho} = \frac{s}{\rho}\;,
\]
we deduce that
\begin{equation}\label{eq:ROAM-generic}
\frac{\E_U[\RAOM(p,y)]}{\pmax}
\geq \con \E[e^{-\rho U} \indic{U \geq \frac{s}{\rho}}] + \rob \E[e^{\rho U} \indic{U < \frac{s}{\rho}}]\;.
\end{equation}
Assume that $\pmax \in [\frac{e^{-\rho}}{\rob}, \frac{1}{\rob}]$, i.e. $s \in [0,\rho]$. The two terms on the right-hand side above can be computed easily
\begin{align*}
\E[e^{-\rho U} \indic{U \geq \frac{s}{\rho}}]
&= \int_{s/\rho}^1 e^{-\rho u} \diff u
= \left[ \frac{-e^{-\rho u}}{\rho} \right]_{s/\rho}^1
= \frac{e^{-s} - e^{-\rho}}{\rho}\;, \\
\E[e^{\rho U} \indic{U < \frac{s}{\rho}}]
&= \int_0^{s/\rho} e^{\rho u} \diff u
= \left[ \frac{e^{\rho u}}{\rho}\right]_0^{s/\rho} 
= \frac{e^s - 1}{\rho}\;,
\end{align*}
and we obtain by substituting into \eqref{eq:ROAM-generic} that
\begin{align}
\frac{\E_U[\RAOM(p,y)]}{\pmax}
&\geq \left( \frac{e^{-s} - e^{-\rho}}{\rho} \right) \con + \left( \frac{e^s - 1}{\rho} \right) \rob \label{eq:RAOM-lb-0<s<rho}\\
&= \left( \frac{1 - e^{-\rho}}{\rho} \right) \con +  \left( \frac{1-e^s}{\rho} \right) \con e^{-s} + \left( \frac{e^s - 1}{\rho} \right) \rob \nonumber\\
&= \left( \frac{1 - e^{-\rho}}{\rho} \right) \con -   \frac{\con e^{-s} - \rob}{\rho} (e^s-1)\nonumber\\
&\geq \left( \frac{1 - e^{-\rho}}{\rho} \right) \con -   \frac{\con - \rob}{\rho} (e^s-1)\;. \nonumber
\end{align}
By definition of $s$, we have
\[
e^s - 1 = \frac{1}{\rob \pmax} - 1
= \frac{1/\rob - \pmax}{\pmax}
\leq \frac{y-\pmax}{\pmax}
= \frac{\err}{\pmax}\;,
\]
which yields 
\[
\frac{\E_U[\RAOM(p,y)]}{\pmax}
\geq  \left( \frac{1 - e^{-\rho}}{\rho} \right) \con -   \left(\frac{\con - \rob}{\rho}\right) \frac{\err}{\pmax}\;.
\]
which corresponds to the consistency/smoothness bound stated in the theorem. However, we only proved it for $\pmax \in [\frac{e^{-\rho}}{r}, \frac{1}{r}]$. We demonstrate in the following that the bound remains true if $\pmax$ is outside that interval.   
If $\pmax \geq 1/\rob$, i.e. $s \leq 0$, then \eqref{eq:ROAM-generic} gives
\begin{align}
\frac{\E_U[\RAOM(p,y)]}{\pmax}
&\geq  \con \E[e^{-\rho U}]
= \left( \frac{1-e^{-\rho}}{\rho} \right) \con \label{eq:RAOM-lb-s<0}\\
&\geq \left( \frac{1 - e^{-\rho}}{\rho} \right) \con -   \left(\frac{\con - \rob}{\rho}\right) \frac{\err}{\pmax}\;. \nonumber
\end{align}
On the other hand, if $\pmax \leq e^{-\rho}/\rob$, i.e. $s \geq \rho$, then \eqref{eq:ROAM-generic} again gives
\begin{equation}\label{eq:RAOM-lb-s>rho}
\frac{\E_U[\RAOM(p,y)]}{\pmax}
\geq  \rob \E[e^{\rho U} \indic{U <1}]
= \left( \frac{e^{\rho}-1}{\rho} \right) \rob\;,
\end{equation}
and it holds that
\[
\frac{\err}{\pmax} = \frac{y}{\pmax} - 1 \geq \frac{1/\rob}{e^{-\rho}/\rob} - 1
= e^{\rho}-1
\geq 1 - e^{-\rho}\;,
\]
hence
\begin{align*}
\frac{\E_U[\RAOM(p,y)]}{\pmax}
&\geq \left( \frac{e^{\rho}-1}{\rho} \right) \rob\\
&\geq  \left( \frac{1-e^{-\rho}}{\rho} \right) \rob\\
&= \left( \frac{1-e^{-\rho}}{\rho} \right) \con - \left(\frac{\con-\rob}{\rho}\right)(1-e^{-\rho})\\
&\geq \left( \frac{1 - e^{-\rho}}{\rho} \right) \con -   \left(\frac{\con - \rob}{\rho}\right) \frac{\err}{\pmax}\;.
\end{align*}
The claimed lower bound is therefore true for all values of $\pmax$.
\end{proof}

\begin{lemma}[Robustness]\label{lem:RAOM-rob}
For any sequence of prices $p$, if $y \in [\frac{1}{\rob},\theta]$, then
\[
\frac{\E_U[\RAOM(p,y)]}{\pmax} \geq \left( \frac{1-e^{-\rho}}{\rho}\right) \rob\;.
\]
\end{lemma}

\begin{proof}
Consider a sequence of prices $p$. Using Inequality \eqref{eq:RAOM-lb-s<0} from the proof of Lemma \ref{lem:RAOM-con}, if $\pmax \geq 1/\rob$ then
\[
\frac{\E_U[\RAOM(p,y)]}{\pmax}
\geq \left( \frac{1-e^{-\rho}}{\rho}\right) \con
\geq \left( \frac{1-e^{-\rho}}{\rho}\right) \rob\;,
\]
and if $\pmax \leq e^{-\rho}/\rob$, then by Inequality \eqref{eq:RAOM-lb-s>rho} we have
\[
\frac{\E_U[\RAOM(p,y)]}{\pmax}
\geq \left( \frac{1-e^{-\rho}}{\rho}\right) \rob\;.
\]
It remains to prove the same lower bound when $\pmax \in (\frac{e^{-\rho}}{\rob}$. Assume that that is the case, and let $s = -\ln(\rob \pmax) \in [0,\rho]$. Inequality \eqref{eq:RAOM-lb-0<s<rho} gives that
\begin{align*}
\frac{\E_U[\RAOM(p,y)]}{\pmax}
&\geq \left( \frac{e^{-s}-e^{-\rho}}{\rho} \right) \con + \left( \frac{e^s-1}{\rho}\right) \rob\\
&= \frac{\rob}{\rho}\left( \frac{\con}{\rob} e^{-s} + e^s - \frac{\con}{\rob}e^{-\rho}-1 \right)\;.
\end{align*}
The function $s \mapsto \frac{\con}{\rob} e^{-s} + e^s$ is minimal on $\R$ for $e^s = \sqrt{\frac{\con}{\rob}}$. Therefore, on the interval $[0,\rho]$, it is minimal for $e^s = \min(e^\rho, \sqrt{\frac{\con}{\rob}})$. If $\sqrt{\frac{\con}{\rob}} \geq e^\rho$ then 
\begin{align*}
\frac{\E_U[\RAOM(p,y)]}{\pmax}
&\geq \frac{\rob}{\rho}\left( \frac{\con}{\rob} e^{-\rho} + e^\rho - \frac{\con}{\rob}e^{-\rho}-1 \right)\\
&= \left( \frac{e^{\rho}-1}{\rho}\right) \rob\\
&\geq \left( \frac{1-e^{-\rho}}{\rho}\right) \rob\;.
\end{align*}
On the other hand, if $\sqrt{\frac{\con}{\rob}} < e^\rho$ then
\begin{align*}
\frac{\E_U[\RAOM(p,y)]}{\pmax}
&\geq \frac{\rob}{\rho}\left( 2 \sqrt{\frac{\con}{\rob}} - \frac{\con}{\rob}e^{-\rho}-1 \right)\;.
\end{align*}
The function $u \mapsto -e^{-\rho} u^2 + 2u - 1$ is non-decreasing on $[1,e^\rho]$, hence, since $\sqrt{\con/\rob \in [1,e^\rho]}$ then
\begin{align*}
\frac{\E_U[\RAOM(p,y)]}{\pmax}
&\geq \frac{\rob}{\rho}\left( 2 - e^{-\rho}-1 \right)
= \left( \frac{1 - e^{-\rho}}{\rho}\right) \rob\;,
\end{align*}
which concludes the proof.

\end{proof}

\section{Average-Case Analysis in Ski-Rental}

In this section, we focus on \textit{ski-rental}, which is one of the fundamental problems in competitive analysis. In this problem, the decision-maker must choose each day between renting a ski for a unit cost or buying it for a fixed cost $b$, allowing them to ski for free for the remainder of the ski season, which has an unknown duration $x$. The objective is to minimize the total cost incurred from renting and buying. 
To simplify our presentation, we consider the continuous version of the problem, where the number of skiing days increases continuously, with $x, b > 0$. In this model, the cost of renting for a time period $[t, t + \delta)$ is equal to $\delta$.

The ski-rental problem was one of the first problems studied in the learning-augmented framework. \citet{purohit2018improving} proved that, with a prediction $y$ of $x$, there is a deterministic algorithm with a competitive ratio of at most
\[
\min\left( 1+\frac{1}{\lambda} \; , \; (1+\lambda) + \frac{|x-y|}{(1-\lambda)\min(x,b)} \right)\;.
\]
where $\lambda \in [0,1]$. It was proved later in \cite{wei2020optimal} that the consistency $(1+\lambda)$ and robustness $(1+\frac{1}{\lambda})$ are Pareto-optimal. 
On the other hand, \cite{benomar2023advice} analyzed the same algorithm under the assumption that $\Pr(\indic{y \geq b} = \indic{x \geq b}) = \p$ for some $\p \in [1/2,1]$, and showed how to optimally choose $\lambda$ to minimize the expected cost of the algorithm.

In the following, we combine the analysis of average-case performance with the criteria of consistency, robustness, and smoothness. To achieve this, we propose a modified version $\ASR$ of the algorithm introduced by \cite{purohit2018improving}, which is parameterized by two parameters $\lambda, \rho \in [0,1]$.

\begin{algorithm}[h]
\caption{$\ASR(x,y)$}\label{algo:ski}
\SetKwInput{Input}{Input}
   \SetKwInOut{Output}{Output}
   \lIf{$y \geq b$}{
    buy at time $\lambda b$
   } \lIf{$y < b$}{
    buy at time $(1+\rho(\tfrac{1}{\lambda}-1))b$
   }
\end{algorithm}

Note that the algorithm of \cite{purohit2018improving} corresponds to $\ASR$ with $\rho = 1$, i.e. buying at time $b/\lambda$ if $y < b$. We start by proving the consistency, robustness, and smoothness of this algorithm.

\begin{theorem}\label{thm:sr-smoothness}
For all $x,y > 0$, denoting by $\eta = |x-y|$, it holds that $\frac{\ASR(x,y)}{\min(x,b)}$ is at most
\[
\min\left( 1+\frac{1}{\lambda} \; , \; (1+\lambda) + \left(1+\frac{\lambda}{\rho}\right) \frac{\err}{\min(x,b)} \right)\;.
\]
\end{theorem}

The theorem above demonstrates that, for any value of $\rho \in [0, 1]$, the algorithm $\ASR$ achieves Pareto-optimal consistency and robustness, albeit with varying levels of smoothness. Furthermore, note that our analysis is tighter than that of \cite{purohit2018improving}. Specifically, when $\rho = 1$, we obtain a smoothness factor of $1+\lambda$ instead of $\frac{1}{1 - \lambda}$.

\begin{proof}
For simplicity, let us denote by $\beta = (1+\rho(\tfrac{1}{\lambda}-1))$. Note that $1 \leq \beta \leq \frac{1}{\lambda}$, and recall that $\min(x,b) = \min(x,b)$.

\noindent
\textbf{Robustness.} We first prove the robustness bound. If $y \geq b$:
\begin{itemize}
    \item if $x < \lambda b$ then $\ASR(x,y) = x = \min(x,b)$,
    \item if $\lambda b \leq x < b$ then $\ASR(x,y) = (1+\lambda)b \leq (1 + \tfrac{1}{\lambda}) x = (1 + \tfrac{1}{\lambda}) \min(x,b)$,
    \item if $b \leq x$, then $\ASR(x,y) = (1+\lambda)b = (1+\lambda)\min(x,b) \leq (1 + \tfrac{1}{\lambda}) \min(x,b)$.
\end{itemize}
On the other hand, if $y< b$:
\begin{itemize}
    \item if $x < b$ then $\ASR(x,y) = x = \min(x,b)$,
    \item if $b \leq x < \beta b$ then $\ASR(x,y) = x < \beta b = \beta \min(x,b) \leq (1 + \tfrac{1}{\lambda}) \min(x,b)$,
    \item if $x \geq \beta x$ then $\ASR(x,y) = (1+\beta) b = (1+\beta)\min(x,b) \leq (1 + \tfrac{1}{\lambda}) \min(x,b)$.
\end{itemize}
In all the cases, it always holds that $\ASR(x,y) \leq (1 + \tfrac{1}{\lambda}) \min(x,b)$.

\noindent
\textbf{Consistency/Smoothness.} Let us first consider the case of $y \geq b$.
\begin{itemize}
    \item if $x < \lambda b$ then $\ASR(x,y) = x = \min(x,b)$,
    \item if $\lambda b \leq x < b$ then $\ASR(x,y) = (1+\lambda)b \leq (1+\lambda)y \leq (1+\lambda)\min(x,b) + (1+\lambda)\err$,
    \item if $b \leq x$, then $\ASR(x,y) = (1+\lambda)b = (1+\lambda)\min(x,b)$.
\end{itemize}
In the case of $y < b$, we obtain that
\begin{itemize}
    \item if $x < b$ then $\ASR(x,y) = x = \min(x,b)$,
    \item if $b \leq x < \beta b$ then \begin{align*}
    \ASR(x,y) 
    &= x \leq y+\err \\
    &\leq b+\err
    = (1+\lambda)b - \lambda b + \err\\
    &\leq (1+\lambda) \min(x,b) + (1-\tfrac{\lambda}{\beta})\err\\
    &\leq (1+\lambda)\min(x,b) + \frac{\beta-\lambda}{\beta - 1} \err\;
    \end{align*}
    where we used in the penultimate inequality that $\err = x-y \leq x \leq \beta b$.
    \item if $x \geq \beta x$ then
    \begin{align*}
    \ASR(x,y) 
    &= (1+\beta) b
    = (1+\lambda)b + (\beta-\lambda) b \\
    &=  (1+\lambda)\min(x,b) + (\beta-\lambda) b \\
    &\leq (1+\lambda)\min(x,b) + \frac{\beta-\lambda}{\beta - 1} \err\;,
    \end{align*}
    where we used in the last inequality that $\err = x-y \geq (\beta-1) b$.
\end{itemize}
All in all, we deduce that
\begin{align*}
\forall x,y: \quad \ASR(x,y)
\leq (1+\lambda)\min(x,b) + \max \left(1+\lambda, \frac{\beta-\lambda}{\beta-1} \right) \err\;,
\end{align*}
and by definition of $\beta$ we have
\[
\frac{\beta - \lambda}{\beta - 1}
= \frac{(\rho + \lambda)(\tfrac{1}{\lambda}-1)}{\rho(\tfrac{1}{\lambda}-1)}
= 1 + \frac{\lambda}{\rho}\;,
\]
hence
\[
\ASR(x,y) \leq (1+\lambda)\min(x,b) + \left(1+\frac{\lambda}{\rho}\right) \err\;,
\]
which concludes the proof.
\end{proof}

In the subsequent theorem, we assume that the prediction $y$ lies on the same side of $b$ as $x$ with a probability of at least $\p \in \left[\frac{1}{2}, 1\right]$, and we establish an upper bound on the expected cost of Algorithm \ref{algo:ski}. The assumption on $y$ is pertinent for this setting, as the decision made by the algorithm depends only on where $y$ is situated compared to $b$. The same assumption was considered in \cite{benomar2023advice}.

\begin{theorem}\label{thm:sr-avg}
For all $x>0$, if the prediction $y$ is a random variable satisfying $\Pr(\indic{y \geq b} = \indic{x\geq b}) \geq \p$ for some $\p \in [\frac{1}{2},1]$, then $\frac{\E_y[\ASR(x,y)]}{\min(x,b)}$ is at most
\[
\max\left(2 + (\tfrac{1}{\lambda}-1)\big((1-\p)\rho -\p\lambda\big), 1 + \frac{1-\p}{\lambda}\right)\;.
\]
\end{theorem}

Note that the upper bound above is non-decreasing with $\rho$, regardless of the values of $\p$ or $\lambda$. Hence, $\rho=0$ is the optimal choice for achieving the best average-case performance of the algorithm. Moreover, if the value of $\p$ is known, then $\lambda$ can also be chosen optimally to minimize the upper bound, as demonstrated in the following corollary.

\begin{proof}
Let us denote by $\beta = 1 + \rho(\tfrac{1}{\lambda} - 1)$, and assume that $\Pr(\indic{y \geq b} = \indic{x\geq b}) = \p$. If $x \geq b$, then 
\begin{itemize}
    \item with probability $\p$: $y \geq b$ and
    \[
    \ASR(x,y) = (1+\lambda)b = (1 + \lambda) \min(x,b)\;,
    \]
    \item with probability $1-\p$: $y<b$ and
    \begin{itemize}
        \item if $x < \beta b$ then $\ASR(x,y) = x \leq \beta b = \beta \min(x,b)$,
        \item if $x \geq \beta b$ then $\ASR(x,y) = (1+\beta) b = (1+\beta) \min(x,b)$,
    \end{itemize}
\end{itemize}
hence we have for $x \geq b$ that
\begin{align}
\frac{\E_y[\ASR(x,y)]}{\min(x,b)} 
&\leq \p(1+\lambda) + (1-\p)(1+\beta) \nonumber\\
&= 1+ \lambda \p + \beta(1-\p) \nonumber\\
&= 1 + \lambda \p + (1 + (\tfrac{1}{\lambda}-1)\rho)(1-\p) \nonumber\\
&= 2 + (\tfrac{1}{\lambda}-1)((1-\p)\rho -\p\lambda)\;. \label{eq:avg-x>b}
\end{align}
On the other hand, for $x < b$
\begin{itemize}
    \item with probability $\p$: $y<b$ and $\ASR(x,y) = x = \min(x,b)$
    \item with probability $1-\p$: $y>b$ and 
    \begin{itemize}
        \item if $x < \lambda b$ then $\ASR(x,y) = x = \min(x,b)$,
        \item if $x < \lambda b$ then $\ASR(x,y) = (1+\lambda) b \leq (1+\tfrac{1}{\lambda})x = (1+\tfrac{1}{\lambda}) \min(x,b)$,
    \end{itemize}
\end{itemize}
which gives for $x < b$ that
\begin{equation}\label{eq:avg-x<b}
\frac{\E_y[\ASR(x,y)]}{\min(x,b)} 
\leq \p + (1-\p)(1+\tfrac{1}{\lambda})
= 1 + \frac{1-\p}{\lambda}\;.    
\end{equation}
We deduce from \eqref{eq:avg-x>b} and \eqref{eq:avg-x<b} that
\[
\forall x: \quad
\frac{\E_y[\ASR(x,y)]}{\min(x,b)} 
\leq \max\left(2 + (\tfrac{1}{\lambda}-1)\big((1-\p)\rho -\p\lambda\big), 1 + \frac{1-\p}{\lambda}\right)\;.
\]
Finally, observe that the right-hand term is a non-increasing function of $\p$, hence the upper bound holds also if $\Pr(\indic{y \geq b} = \indic{x\geq b}) \geq \p$.
\end{proof}

\begin{corollary}\label{cor:sr-avg}
Under the same assumptions of Theorem \ref{thm:sr-avg}, it holds for $\rho = 0$ and $\lambda^* = \frac{1}{2} \sqrt{(\frac{1}{\p}-1)( \frac{1}{\p} + 3 )} - \tfrac{1}{2}( \frac{1}{\p} - 1 )$ that
\[
\frac{\E_y[\A^{\operatorname{SR}}_{\lambda^*,0}(x,y)]}{\min(x,b)} 
\leq \frac{3-\p}{2} + \frac{1}{2} \sqrt{(1-\p)(1+3\p)}\;.
\]
\end{corollary}

\begin{figure}
    \centering
    \includegraphics[width=0.5\linewidth]{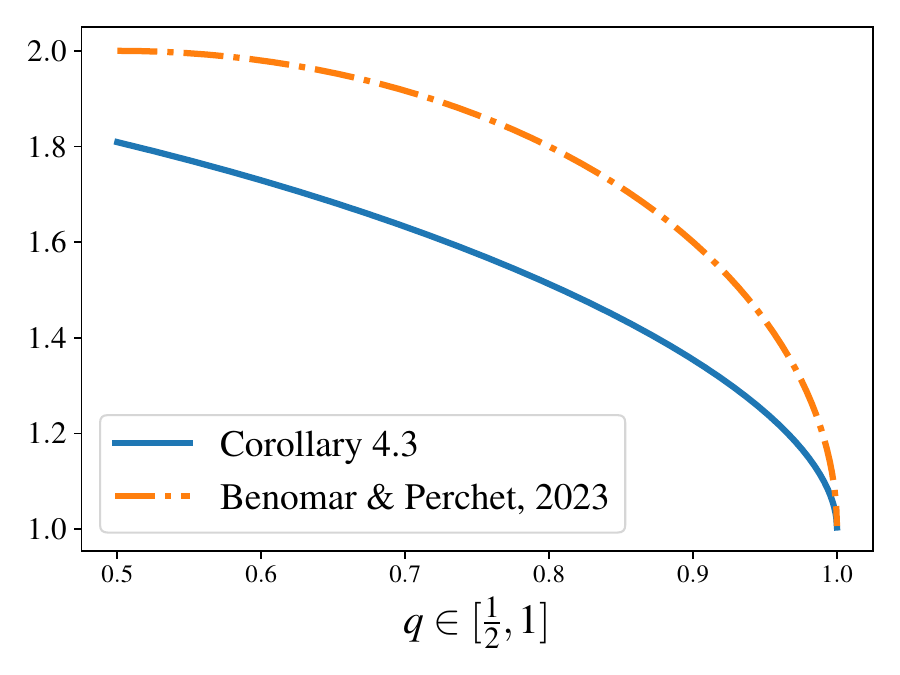}
    \caption{Upper bound on the competitive ratio of $\ASR$ with $\lambda, \rho$ as in Corollary \ref{cor:sr-avg}, and with $\lambda, \rho$ as in Lemma 2.2 of \cite{benomar2023advice}.}
    \label{fig:avg-sr}
\end{figure}

Under the same assumption on the prediction $y$, \citet{benomar2023advice} proved an upper bound of $(1+2\sqrt{\p(1-\p)}) \min(x,b)$ on the average cost of $\A^{\operatorname{SR}}_{\lambda,1}$ for a well-chosen value of $\lambda$. The bound of Corollary \ref{cor:sr-avg} is better than the latter as shown in Figure \ref{fig:avg-sr}.
Moreover, note that for $\p=1/2$, the bound of Corollary \ref{cor:sr-avg} is better than $2$, which is the best competitive ratio achievable by a deterministic algorithm for the ski-rental problem. This is because, in that case, $\ASR$ uses the Bernoulli random variable $\indic{y \geq b} \sim \mathcal{B}(1/2)$.

\begin{proof}
For all $\p \in [\frac{1}{2},1]$ and $\rho \in [0,1]$, the upper bound of Theorem \ref{thm:sr-avg} is non-decreasing with respect to $\rho$, hence the optimal choice of $\rho$ is 0. With $\rho = 0$, let us examine for which value of $\lambda$ the two terms in the maximum of the upper bound in Theorem \ref{thm:sr-avg} are equal.  We have for all $\lambda \in [0,1]$ the equivalences
\begin{align*}
2 - (1-\lambda)\p = 1 + \frac{1-\p}{\lambda}
&\iff 1-(1-\lambda) \p = \frac{1-\p}{\lambda}\\
&\iff \tfrac{\lambda}{\p} - \lambda(1-\lambda) = \tfrac{1}{\p}-1\\
&\iff \lambda^2 + \left( \tfrac{1}{\p} - 1 \right) \lambda - \left(\tfrac{1}{\p} - 1\right) = 0\\
&\iff \lambda = \tfrac{1}{2} \sqrt{\left(\tfrac{1}{\p} - 1\right) \left( 3 + \tfrac{1}{\p}\right)} - \tfrac{1}{2}\left(\tfrac{1}{\p} - 1\right)\;.
\end{align*}
Let us denote by $\lambda^*$ the expression of $\lambda$ above. It holds that $\lambda^* \in [0,1]$. Indeed,
\begin{align*}
\lambda^* 
&= \tfrac{1}{2} \sqrt{\left(\tfrac{1}{\p} - 1\right) \left(\tfrac{1}{\p} + 3\right)} - \tfrac{1}{2}\left(\tfrac{1}{\p} - 1\right)
\geq \tfrac{1}{2} \sqrt{\left(\tfrac{1}{\p} - 1\right)^2} - \tfrac{1}{2}\left(\tfrac{1}{\p} - 1\right) = 0\;,\\
\lambda^*
&= \tfrac{1}{2} \sqrt{\tfrac{1}{\p^2} + \tfrac{2}{\p} - 3} - \tfrac{1}{2}\left(\tfrac{1}{\p} - 1\right)
\leq \tfrac{1}{2} \left(\tfrac{1}{\p} + 1\right) - \tfrac{1}{2}\left(\tfrac{1}{\p} - 1\right) = 1\;.
\end{align*}
Therefore, $\lambda^*$ is a valid value of $\lambda$ that can be chosen by the decision-maker, which yields the following upper bound on the average cost of $\ASR$
\begin{align*}
\frac{\E_y[\A^{\operatorname{SR}}_{\lambda^*,0}(x,y)]}{\min(x,b)} 
&\leq \max\left(2 - (1 - \lambda^*) \p, 1 + \frac{1-\p}{\lambda}\right)\\
&= 2 - (1 - \lambda^*) \p\\
&= 2 - \frac{1+\p}{2} + \frac{\p^2}{2} \sqrt{\left(\tfrac{1}{\p} - 1\right) \left( 3 + \tfrac{1}{\p}\right)}\\
&= \frac{3-\p}{2} + \frac{1}{2} \sqrt{(1-\p)(1+3\p)}\;.
\end{align*}

\end{proof}

\paragraph{Smoothness and Average-Cost Tradeoff.}
A natural tradeoff arises between the consistency, robustness, and average cost of the algorithm. Minimizing the average cost necessitates selecting $\lambda$ optimally to minimize the upper bound established in Theorem \ref{thm:sr-avg}. However, the decision-maker may opt to deviate from this value to achieve a better level of consistency or robustness, depending on the specific use case.

Furthermore, Theorems \ref{thm:sr-smoothness} and \ref{thm:sr-avg} imply that, for a fixed $\lambda \in [0, 1]$, the levels of consistency and robustness of $\ASR$ are constant, while the smoothness and average cost can be further adjusted using $\rho$.
While increasing $\rho$ enhances the smoothness, it degrades the average cost, regardless of the accuracy $\p$ of the predictions. This indicates that, in addition to the tradeoff between consistency and robustness governed by $\lambda$, the algorithm also exhibits a tradeoff between average cost and smoothness, governed by the parameter $\rho$.

\section{EXPERIMENTS}

In this section, we present experimental results to validate our theoretical findings and provide additional insights into the tradeoffs discussed in the paper.

\begin{figure}
    \centering
    \begin{minipage}{0.52\textwidth}
        \centering
        \includegraphics[width=\linewidth]{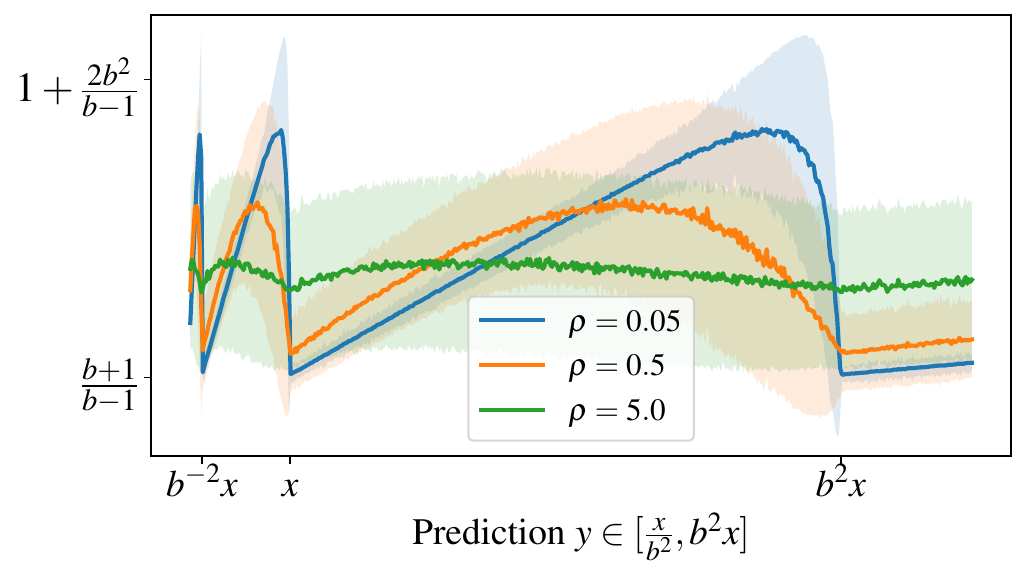}
        \caption{Consistency-smoothness tradeoff of $\AOS$ with a prediction randomized as in Theorem \ref{thm:smooth-LS}}
        \label{fig:line-search-experiment}
    \end{minipage}
    \hfill
    \begin{minipage}{0.44\textwidth}
        \centering
        \includegraphics[width=\linewidth]{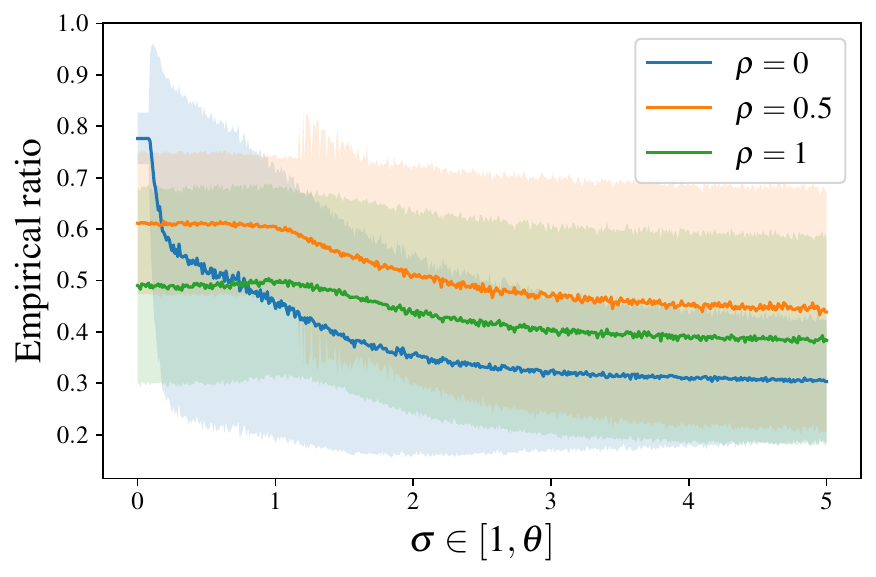}
        \caption{Consistency-smoothness tradeoff of Algorithm $\RAOM$ for $\rho \in \{0,0.5,1\}$.}
        \label{fig:RAOM}
    \end{minipage}
\end{figure}

\paragraph{Line Search.}
As established in Lemma \ref{lem:wc-po-os} and illustrated in Figure \ref{fig:lb-onlineSearch}, given a target position $x \geq 1$ and prediction $y > 0$, the ratio between the distance traveled by Algorithm $\AOS$ and the optimal offline algorithm depends solely on the ratio $x/y$ when $x$ is large. To investigate the impact of the parameter $\rho$, we fix $x = 100$ and $b = 2.5$, then compare the behavior of the algorithm presented in Section \ref{sec:line-search} for three different values of $\rho \in \{0.05, 0.5, 5\}$, with $y \in [\tfrac{x}{b^2}, b^2 x]$. For each point in the experiment, the average and standard deviation are computed over $10^5$ independent trials.
Figure \ref{fig:line-search-experiment} demonstrates that smaller values of $\rho$ lead to better consistency, but make the algorithm highly sensitive to prediction errors. This highlights the consistency-smoothness tradeoff established in Theorem \ref{thm:smooth-LS}.

\paragraph{One-Max Search.} 
We conduct an analogous experiment for the one-max search problem to demonstrate the consistency-smoothness tradeoff for $\RAOM$. Given a sequence with a maximal price $\pmax$, the algorithm is provided with a noisy prediction in the form $y = \pmax + \varepsilon$, where $\varepsilon \sim \mathcal{U}[-\sigma, \sigma]$. The threshold set by $\RAOM$, denoted as $\rthresh$, determines its worst-case payoff: if $\pmax \geq \rthresh$, the algorithm gains $\rthresh$; otherwise, the gain is $1$, which is the minimum possible price in the sequence. This scenario is asymptotically achieved by the sequence $p = (p_1, \ldots, p_{n+1})$, where $p_i = 1 + \tfrac{i-1}{n-1}(\pmax - 1)$ for $i \leq n$, and $p_{n+1} = 1$. In the experiment, $\lambda = 0.1$ and $\theta = 5$ are fixed, and for each $\sigma \in [0, \theta]$, the worst-case average ratio $\sup_{\pmax \in [1, \theta]} \E[\RAOM(p, \pmax + \varepsilon)]/\pmax$ and the corresponding standard deviation are evaluated over $10^5$ independent samples.
Figure \ref{fig:RAOM} shows that the algorithm suffers from brittleness for $\rho = 0$, as the slightest prediction error substantially degrades its performance. In contrast, as $\rho$ increases and randomization is introduced, the algorithm becomes smoother; at the cost of consistency.

\paragraph{Ski Rental.}

\begin{figure}
    \centering
    \begin{minipage}{0.48\textwidth}
        \centering
        \includegraphics[width=\linewidth]{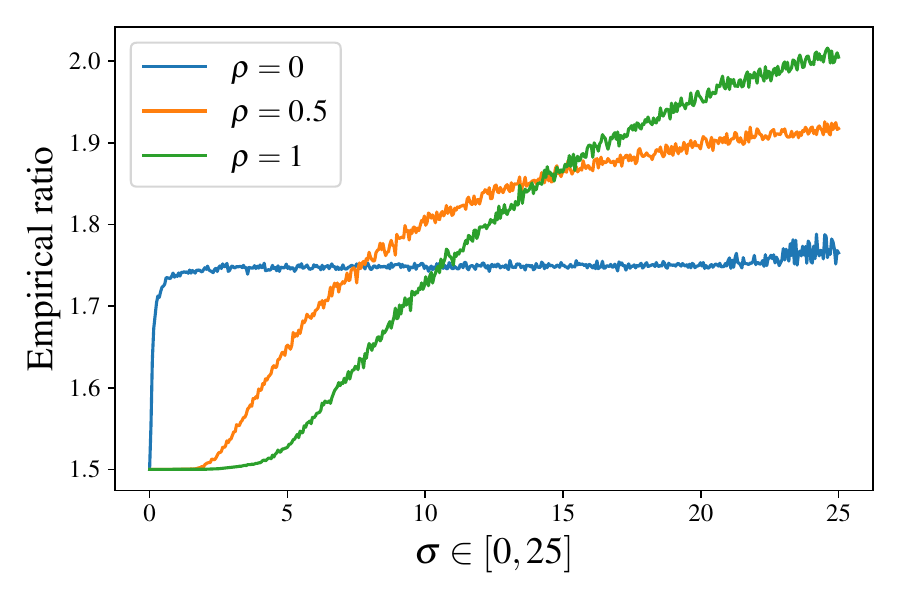}
        \caption{Consistency-smoothness tradeoff of $\ASR$, with $y \sim x + \mathcal{N}(0,\sigma^2)$}
        \label{fig:ASR-gaussian}
    \end{minipage}
    \hfill
    \begin{minipage}{0.48\textwidth}
        \centering
        \includegraphics[width=\linewidth]{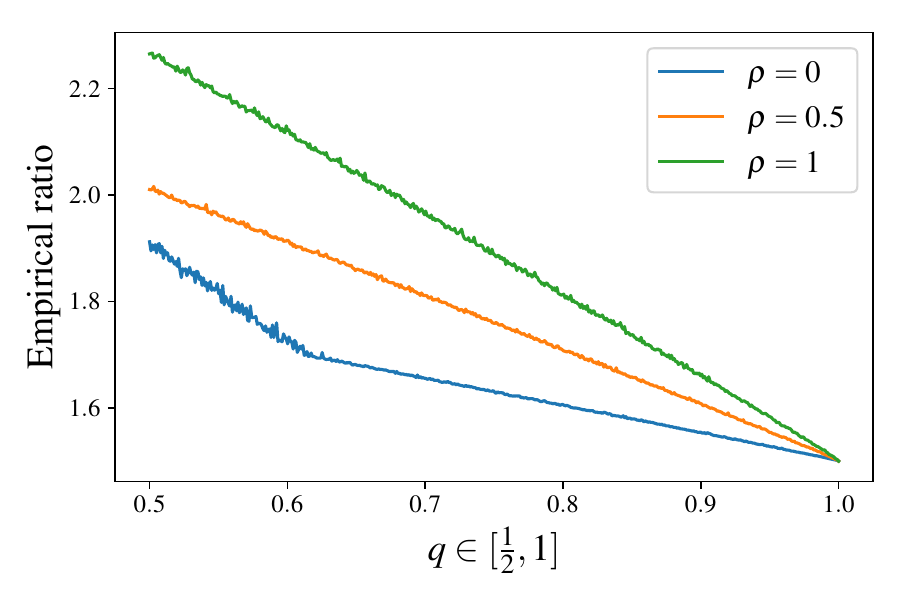}
        \caption{Average-case performance of $\ASR$ when $\indic{y \geq b} = \indic{x \geq b}$ with probability $\p$}
        \label{fig:ASR-bernoulli}
    \end{minipage}
\end{figure}

For the ski-rental problem, two experiments are conducted to investigate, on one hand, the impact of the parameter $\rho$ on the consistency and smoothness of the algorithm $\ASR$, and on the other hand, its impact on the average performance of $\ASR$, under the assumption that the prediction lies on the same side of $b$ as the true value of the number of snow days $x$ with probability $\p$.
In both experiments, we set $b = 10$ and $\lambda = 0.5$.

In Figure \ref{fig:ASR-gaussian}, the performance of $\ASR$ is evaluated against the prediction error $\eta = |y - x|$. Predictions are of the form $y = x + \varepsilon$, where $\varepsilon \sim \mathcal{N}(0, \sigma^2)$. For each value of $\sigma$, the figure shows the worst-case ratio $\sup_{x \in (0,5b]} {\ASR(x, y)}/{\min(x, b)}$, which compares the cost induced by $\ASR$ to that of the optimal offline algorithm at a fixed level of error.

When $\rho = 0$, the algorithm lacks smoothness, resulting in a significant increase in cost from the consistency value of $1 + \lambda = 1.5$ when the error is zero, to $\approx 1.7$ for even a minimal positive error. For larger values of $\rho$, the algorithm exhibits improved smoothness while maintaining the same consistency level, as proved in Theorem \ref{thm:sr-smoothness}. However, it is noteworthy that smaller values of $\rho$ tend to yield a better average ratio when the prediction error is large.

On the other hand, Figure \ref{fig:ASR-bernoulli} examines the the average cost of $\ASR$, assuming that $\indic{y \geq b} = \indic{x \geq b}$ with probability $\p \in \left[\frac{1}{2}, 1\right]$. Each data point in the figure is computed over $10^5$ independent trials, where $x \sim \mathcal{U}[1, 4b]$, and $y$ is selected arbitrarily on the same side of $b$ as $x$ with probability $\p$, while on the opposite side with probability $1-\p$. The results further corroborate our theoretical findings, indicating that smaller values of $\rho$ yield the best average performance.

Combined with the observations from Figure \ref{fig:ASR-gaussian}, these experiments delineate the tradeoff between the smoothness and average-case performance of $\ASR$, which can be tuned using the parameter $\rho$.

\bibliographystyle{plainnat}
\bibliography{bibliography}

\begin{thebibliography}{41}
\providecommand{\natexlab}[1]{#1}
\providecommand{\url}[1]{\texttt{#1}}
\expandafter\ifx\csname urlstyle\endcsname\relax
  \providecommand{\doi}[1]{doi: #1}\else
  \providecommand{\doi}{doi: \begingroup \urlstyle{rm}\Url}\fi

\bibitem[Angelopoulos(2023)]{angelopoulos2023online}
Spyros Angelopoulos.
\newblock Online search with a hint.
\newblock \emph{Information and Computation}, 295:\penalty0 105091, 2023.

\bibitem[Angelopoulos et~al.(2019)Angelopoulos, D{\"u}rr, Jin, Kamali, and Renault]{angelopoulos2019online}
Spyros Angelopoulos, Christoph D{\"u}rr, Shendan Jin, Shahin Kamali, and Marc Renault.
\newblock Online computation with untrusted advice.
\newblock \emph{arXiv preprint arXiv:1905.05655}, 2019.

\bibitem[Angelopoulos et~al.(2024{\natexlab{a}})Angelopoulos, D{\"u}rr, Elenter, and Lefki]{angelopoulos2024overcoming}
Spyros Angelopoulos, Christoph D{\"u}rr, Alex Elenter, and Yanni Lefki.
\newblock Overcoming brittleness in pareto-optimal learning-augmented algorithms.
\newblock \emph{arXiv preprint arXiv:2408.04122}, 2024{\natexlab{a}}.

\bibitem[Angelopoulos et~al.(2024{\natexlab{b}})Angelopoulos, Lidbetter, and Panagiotou]{angelopoulos2024search}
Spyros Angelopoulos, Thomas Lidbetter, and Konstantinos Panagiotou.
\newblock Search games with predictions.
\newblock \emph{arXiv preprint arXiv:2401.01149}, 2024{\natexlab{b}}.

\bibitem[Antoniadis et~al.(2021)Antoniadis, Coester, Eli{\'a}s, Polak, and Simon]{antoniadis2021learning}
Antonios Antoniadis, Christian Coester, Marek Eli{\'a}s, Adam Polak, and Bertrand Simon.
\newblock Learning-augmented dynamic power management with multiple states via new ski rental bounds.
\newblock \emph{Advances in Neural Information Processing Systems}, 34:\penalty0 16714--16726, 2021.

\bibitem[Antoniadis et~al.(2023{\natexlab{a}})Antoniadis, Boyar, Eli{\'a}s, Favrholdt, Hoeksma, Larsen, Polak, and Simon]{antoniadis2023paging}
Antonios Antoniadis, Joan Boyar, Marek Eli{\'a}s, Lene~Monrad Favrholdt, Ruben Hoeksma, Kim~S Larsen, Adam Polak, and Bertrand Simon.
\newblock Paging with succinct predictions.
\newblock In \emph{International Conference on Machine Learning}, pages 952--968. PMLR, 2023{\natexlab{a}}.

\bibitem[Antoniadis et~al.(2023{\natexlab{b}})Antoniadis, Coester, Eli{\'a}{\v{s}}, Polak, and Simon]{antoniadis2023online}
Antonios Antoniadis, Christian Coester, Marek Eli{\'a}{\v{s}}, Adam Polak, and Bertrand Simon.
\newblock Online metric algorithms with untrusted predictions.
\newblock \emph{ACM Transactions on Algorithms}, 19\penalty0 (2):\penalty0 1--34, 2023{\natexlab{b}}.

\bibitem[Antunes and Fortnow(2009)]{antunes2009worst}
Luis Antunes and Lance Fortnow.
\newblock Worst-case running times for average-case algorithms.
\newblock In \emph{2009 24th Annual IEEE Conference on Computational Complexity}, pages 298--303. IEEE, 2009.

\bibitem[Baezayates et~al.(1993)Baezayates, Culberson, and Rawlins]{baezayates1993searching}
Ricardo~A Baezayates, Joseph~C Culberson, and Gregory~JE Rawlins.
\newblock Searching in the plane.
\newblock \emph{Information and computation}, 106\penalty0 (2):\penalty0 234--252, 1993.

\bibitem[Bamas et~al.(2020)Bamas, Maggiori, and Svensson]{bamas2020primal}
Etienne Bamas, Andreas Maggiori, and Ola Svensson.
\newblock The primal-dual method for learning augmented algorithms.
\newblock \emph{Advances in Neural Information Processing Systems}, 33:\penalty0 20083--20094, 2020.

\bibitem[Beck(1964)]{beck1964linear}
Anatole Beck.
\newblock On the linear search problem.
\newblock \emph{Israel Journal of Mathematics}, 2\penalty0 (4):\penalty0 221--228, 1964.

\bibitem[Beck and Newman(1970)]{beck1970yet}
Anatole Beck and Donald~J Newman.
\newblock Yet more on the linear search problem.
\newblock \emph{Israel journal of mathematics}, 8\penalty0 (4):\penalty0 419--429, 1970.

\bibitem[Benomar and Coester(2024)]{benomar2024learning}
Ziyad Benomar and Christian Coester.
\newblock Learning-augmented priority queues.
\newblock \emph{arXiv preprint arXiv:2406.04793}, 2024.

\bibitem[Benomar and Perchet(2023)]{benomar2023advice}
Ziyad Benomar and Vianney Perchet.
\newblock Advice querying under budget constraint for online algorithms.
\newblock In \emph{Thirty-seventh Conference on Neural Information Processing Systems}, 2023.

\bibitem[Benomar and Perchet(2024)]{benomarnon}
Ziyad Benomar and Vianney Perchet.
\newblock Non-clairvoyant scheduling with partial predictions.
\newblock In \emph{Forty-first International Conference on Machine Learning}, 2024.

\bibitem[Ch{l}{k{e}}dowski et~al.(2021)Ch{l}{k{e}}dowski, Polak, Szabucki, and {.Z}o{l}na]{chlkedowski2021robust}
Jakub Ch{l}{k{e}}dowski, Adam Polak, Bartosz Szabucki, and Konrad~Tomasz {.Z}o{l}na.
\newblock Robust learning-augmented caching: An experimental study.
\newblock In \emph{International Conference on Machine Learning}, pages 1920--1930. PMLR, 2021.

\bibitem[Chuangpishit et~al.(2018)Chuangpishit, Georgiou, and Sharma]{chuangpishit2018average}
Huda Chuangpishit, Konstantinos Georgiou, and Preeti Sharma.
\newblock Average case-worst case tradeoffs for evacuating 2 robots from the disk in the face-to-face model.
\newblock In \emph{International Symposium on Algorithms and Experiments for Sensor Systems, Wireless Networks and Distributed Robotics}, pages 62--82. Springer, 2018.

\bibitem[Cohen-Addad et~al.(2024)Cohen-Addad, d'Orsi, Gupta, Lee, and Panigrahi]{cohen2024max}
Vincent Cohen-Addad, Tommaso d'Orsi, Anupam Gupta, Euiwoong Lee, and Debmalya Panigrahi.
\newblock Max-cut with $\varepsilon$-accurate predictions.
\newblock \emph{arXiv preprint arXiv:2402.18263}, 2024.

\bibitem[Diakonikolas et~al.(2021)Diakonikolas, Kontonis, Tzamos, Vakilian, and Zarifis]{diakonikolas2021learning}
Ilias Diakonikolas, Vasilis Kontonis, Christos Tzamos, Ali Vakilian, and Nikos Zarifis.
\newblock Learning online algorithms with distributional advice.
\newblock In \emph{International Conference on Machine Learning}, pages 2687--2696. PMLR, 2021.

\bibitem[D{\"u}tting et~al.(2021)D{\"u}tting, Lattanzi, Paes~Leme, and Vassilvitskii]{dutting2021secretaries}
Paul D{\"u}tting, Silvio Lattanzi, Renato Paes~Leme, and Sergei Vassilvitskii.
\newblock Secretaries with advice.
\newblock In \emph{Proceedings of the 22nd ACM Conference on Economics and Computation}, pages 409--429, 2021.

\bibitem[El-Yaniv et~al.(2001)El-Yaniv, Fiat, Karp, and Turpin]{el2001optimal}
Ran El-Yaniv, Amos Fiat, Richard~M Karp, and Gordon Turpin.
\newblock Optimal search and one-way trading online algorithms.
\newblock \emph{Algorithmica}, 30:\penalty0 101--139, 2001.

\bibitem[Gollapudi and Panigrahi(2019)]{gollapudi2019online}
Sreenivas Gollapudi and Debmalya Panigrahi.
\newblock Online algorithms for rent-or-buy with expert advice.
\newblock In \emph{International Conference on Machine Learning}, pages 2319--2327. PMLR, 2019.

\bibitem[Gupta et~al.(2022)Gupta, Panigrahi, Subercaseaux, and Sun]{gupta2022augmenting}
Anupam Gupta, Debmalya Panigrahi, Bernardo Subercaseaux, and Kevin Sun.
\newblock Augmenting online algorithms with $\varepsilon$-accurate predictions.
\newblock \emph{Advances in neural information processing systems}, 35:\penalty0 2115--2127, 2022.

\bibitem[Henzinger et~al.(2023)Henzinger, Saha, Seybold, and Ye]{henzinger2023complexity}
Monika Henzinger, Barna Saha, Martin~P Seybold, and Christopher Ye.
\newblock On the complexity of algorithms with predictions for dynamic graph problems.
\newblock \emph{arXiv preprint arXiv:2307.16771}, 2023.

\bibitem[Karlin et~al.(1988)Karlin, Manasse, Rudolph, and Sleator]{karlin1988competitive}
Anna~R Karlin, Mark~S Manasse, Larry Rudolph, and Daniel~D Sleator.
\newblock Competitive snoopy caching.
\newblock \emph{Algorithmica}, 3:\penalty0 79--119, 1988.

\bibitem[Karlin et~al.(1994)Karlin, Manasse, McGeoch, and Owicki]{karlin1994competitive}
Anna~R. Karlin, Mark~S. Manasse, Lyle~A. McGeoch, and Susan Owicki.
\newblock Competitive randomized algorithms for nonuniform problems.
\newblock \emph{Algorithmica}, 11\penalty0 (6):\penalty0 542--571, 1994.

\bibitem[Kraska et~al.(2018)Kraska, Beutel, Chi, Dean, and Polyzotis]{kraska2018case}
Tim Kraska, Alex Beutel, Ed~H Chi, Jeffrey Dean, and Neoklis Polyzotis.
\newblock The case for learned index structures.
\newblock In \emph{Proceedings of the 2018 international conference on management of data}, pages 489--504, 2018.

\bibitem[Lassota et~al.(2023)Lassota, Lindermayr, Megow, and Schl{\"o}ter]{lassota2023minimalistic}
Alexandra~Anna Lassota, Alexander Lindermayr, Nicole Megow, and Jens Schl{\"o}ter.
\newblock Minimalistic predictions to schedule jobs with online precedence constraints.
\newblock In \emph{International Conference on Machine Learning}, pages 18563--18583. PMLR, 2023.

\bibitem[Lin et~al.(2022)Lin, Luo, and Woodruff]{lin2022learning}
Honghao Lin, Tian Luo, and David Woodruff.
\newblock Learning augmented binary search trees.
\newblock In \emph{International Conference on Machine Learning}, pages 13431--13440. PMLR, 2022.

\bibitem[Lykouris and Vassilvtiskii(2018)]{lykouris2018competitive}
Thodoris Lykouris and Sergei Vassilvtiskii.
\newblock Competitive caching with machine learned advice.
\newblock In \emph{International Conference on Machine Learning}, pages 3296--3305. PMLR, 2018.

\bibitem[Merlis et~al.(2023)Merlis, Richard, Sentenac, Odic, Molina, and Perchet]{merlis2023preemption}
Nadav Merlis, Hugo Richard, Flore Sentenac, Corentin Odic, Mathieu Molina, and Vianney Perchet.
\newblock On preemption and learning in stochastic scheduling.
\newblock In \emph{International Conference on Machine Learning}, pages 24478--24516. PMLR, 2023.

\bibitem[Peikert and Rosen(2007)]{peikert2007lattices}
Chris Peikert and Alon Rosen.
\newblock Lattices that admit logarithmic worst-case to average-case connection factors.
\newblock In \emph{Proceedings of the thirty-ninth annual ACM symposium on Theory of computing}, pages 478--487, 2007.

\bibitem[Purohit et~al.(2018)Purohit, Svitkina, and Kumar]{purohit2018improving}
Manish Purohit, Zoya Svitkina, and Ravi Kumar.
\newblock Improving online algorithms via ml predictions.
\newblock \emph{Advances in Neural Information Processing Systems}, 31, 2018.

\bibitem[Rice et~al.(2021)Rice, Bair, Zhang, and Kolter]{rice2021robustness}
Leslie Rice, Anna Bair, Huan Zhang, and J~Zico Kolter.
\newblock Robustness between the worst and average case.
\newblock \emph{Advances in Neural Information Processing Systems}, 34:\penalty0 27840--27851, 2021.

\bibitem[Robey et~al.(2022)Robey, Chamon, Pappas, and Hassani]{robey2022probabilistically}
Alexander Robey, Luiz Chamon, George~J Pappas, and Hamed Hassani.
\newblock Probabilistically robust learning: Balancing average and worst-case performance.
\newblock In \emph{International Conference on Machine Learning}, pages 18667--18686. PMLR, 2022.

\bibitem[Shin et~al.(2023)Shin, Lee, Lee, and An]{shin2023improved}
Yongho Shin, Changyeol Lee, Gukryeol Lee, and Hyung-Chan An.
\newblock Improved learning-augmented algorithms for the multi-option ski rental problem via best-possible competitive analysis.
\newblock \emph{arXiv preprint arXiv:2302.06832}, 2023.

\bibitem[Sun et~al.(2021)Sun, Lee, Hajiesmaili, Wierman, and Tsang]{sun2021pareto}
Bo~Sun, Russell Lee, Mohammad Hajiesmaili, Adam Wierman, and Danny Tsang.
\newblock Pareto-optimal learning-augmented algorithms for online conversion problems.
\newblock \emph{Advances in Neural Information Processing Systems}, 34:\penalty0 10339--10350, 2021.

\bibitem[Szirmay-Kalos and M{\'a}rton(1998)]{szirmay1998worst}
L{\'a}szl{\'o} Szirmay-Kalos and G{\'a}bor M{\'a}rton.
\newblock Worst-case versus average case complexity of ray-shooting.
\newblock \emph{Computing}, 61\penalty0 (2):\penalty0 103--131, 1998.

\bibitem[Wei and Zhang(2020)]{wei2020optimal}
Alexander Wei and Fred Zhang.
\newblock Optimal robustness-consistency trade-offs for learning-augmented online algorithms.
\newblock \emph{Advances in Neural Information Processing Systems}, 33:\penalty0 8042--8053, 2020.

\bibitem[Witt(2005)]{witt2005worst}
Carsten Witt.
\newblock Worst-case and average-case approximations by simple randomized search heuristics.
\newblock In \emph{Annual Symposium on Theoretical Aspects of Computer Science}, pages 44--56. Springer, 2005.

\bibitem[Zeynali et~al.()Zeynali, Kamali, and Hajiesmaili]{zeynalirobust}
Ali Zeynali, Shahin Kamali, and Mohammad Hajiesmaili.
\newblock Robust learning-augmented dictionaries.
\newblock In \emph{Forty-first International Conference on Machine Learning}.

\end{thebibliography}


\end{document}